\newif\ifsubmit

\newif\ifblind      
\newif\ifcompact    
\newif\ifexabs      
\newif\ifitcs       
\newif\ifllncs

\submittrue

\ifllncs
\documentclass[runningheads,a4paper]{llncs}

\else
\documentclass[11pt,pdfa]{article}
\usepackage[in]{fullpage}
\fi

\usepackage{iftex}
\ifPDFTeX
  \usepackage[utf8]{inputenc}
  \usepackage[noTeX]{mmap}
  \usepackage[T1]{fontenc}
  
\fi
\ifLuaTeX
  \usepackage{luatex85}
  \usepackage[noTeX]{mmap}
\fi

\usepackage[style=alphabetic,minalphanames=3,maxalphanames=4,maxnames=99,backref=true]{biblatex}

\usepackage[normalem]{ulem}
\usepackage{comment}
\usepackage[shortlabels]{enumitem}
\usepackage{bm}
\usepackage{amsfonts}
\usepackage{amsmath}
\usepackage{amssymb}
\usepackage{amsthm}
\usepackage{xcolor}
\usepackage{graphicx}
\usepackage{qtree}
\usepackage{tree-dvips}
\usepackage{float}
\usepackage{hyperref}
\usepackage[nameinlink]{cleveref}
\usepackage{braket}
\usepackage{mathrsfs}
\usepackage{tikz}

  \hypersetup{colorlinks={true},linkcolor={blue},citecolor=magenta}

\ifllncs
\else
  \theoremstyle{plain}
  \newtheorem{theorem}{Theorem}[section]
  \newtheorem{lemma}[theorem]{Lemma}
  \newtheorem{claim}{Claim}
  \newtheorem{corollary}[theorem]{Corollary}

  \newtheorem{definition}[theorem]{Definition}
  \newtheorem{remark}[theorem]{Remark}
  \newtheorem{conjecture}[theorem]{Conjecture}
    \newtheorem{fact}[theorem]{Fact}

\fi
  
  \usepackage[style=alphabetic,minalphanames=3,maxalphanames=4,maxnames=99,backref=true]{biblatex}

\DeclareUnicodeCharacter{0301}{\'{e}}

\newcommand{\ketbra}[2]{| #1 \rangle \langle #2 |}
\newcommand{\poly}{\mathrm{poly}}

\newcommand{\calS}{{\cal S}}
\newcommand{\td}{{\sf TD}}

\newcommand{\parties}{{\cal P}}

\newcommand{\ct}{{\sf ct}}
\newcommand{\sk}{{\sf sk}}
\newcommand{\share}{{\sf Share}}
\newcommand{\enc}{{\sf Enc}}

\newcommand{\keygen}{{\sf KeyGen}}
\newcommand{\dec}{{\sf Dec}}
\newcommand{\cnot}{{\sf CNOT}}
\usepackage{soul}

\newcommand{\Tr}{\mathsf{Tr}}

\newcommand{\reconstruct}{{\sf Reconstruct}}
\newcommand{\alice}{{\cal A}}
\newcommand{\bob}{{\cal B}}
\newcommand{\charlie}{{\cal C}}
\newcommand{\As}{{\cal A}}
\newcommand{\Bs}{{\cal B}}
\newcommand{\Cs}{{\cal C}}
\newcommand{\adversary}{{\sf Adv}}
\newcommand{\regX}{\mathbf{X}}
\newcommand{\regY}{\mathbf{Y}}
\newcommand{\secparam}{\lambda}
\newcommand{\prob}{\mathsf{Pr}}
\newcommand{\negl}{\mathsf{negl}}
\newcommand{\expt}{\mathsf{Expt}}
\newcommand{\aux}{\mathbf{Aux}}
\newcommand{\cP}{{\cal P}}
\newcommand{\I}{\mathbb{I}}

\ifsubmit
\newcommand{\pnote}[1]{}
\newcommand{\qipeng}[1]{}
\newcommand{\vipul}[1]{}
\newcommand{\vnote}[1]{}
\newcommand{\jiahui}[1]{}

\else
\newcommand{\pnote}[1]{{\color{red} P:#1}}
\newcommand{\qipeng}[1]{{\color{red} Q:#1}}
\newcommand{\vipul}[1]{{\color{purple} V:#1}}
\newcommand{\vnote}[1]{{\color{purple} V:#1}}
\newcommand{\jiahui}[1]{{\color{blue} J:#1}}
\fi

\newcommand{\T}{{\sf T}}

\newcommand{\X}{{\sf X}}
\newcommand{\Z}{{\sf Z}}
\newcommand{\Pgate}{{\sf P}}
\newcommand{\Hgate}{{\sf H}}
\newcommand{\CNOT}{{\sf CNOT}}

\newcommand{\ans}{{\sf ans}}

\usepackage[style=alphabetic,minalphanames=3,maxalphanames=4,maxnames=99,backref=true]{biblatex}

  \DeclareFieldFormat{eprint:iacr}{Cryptology ePrint Archive: \href{https://ia.cr/#1}{\texttt{#1}}}
  \DeclareFieldFormat{eprint:iacrarchive}{Cryptology ePrint Archive: \href{https://eprint.iacr.org/archive/#1}{\texttt{#1}}}
\title{Unclonable Secret Sharing}
\author{
Prabhanjan Ananth\thanks{University of California, Santa Barbara. \texttt{prabhanjan@cs.ucsb.edu}. }
\and Vipul Goyal\thanks{NTT Research, Carnegie Mellon University. \texttt{vipul@cmu.edu}.}
    \and Jiahui Liu\thanks{{Massachusetts Institute of Technology.}
\texttt{{jiahuiliu@csail.mit.edu}}.}
\and Qipeng Liu\thanks{University of California, San Diego. \texttt{qipengliu0@gmail.com}.}
}

\date{}

\begin{document}

\maketitle

\newcommand{\USS}{{\sf USS}}
\newcommand{\uss}{{\sf{USS}}}

\newcommand{\UE}{{\sf UE}}
\newcommand{\adv}{\mathcal{A}}
\newcommand{\cO}{\mathcal{O}}

\newcommand{\advantage}{\mathsf{adv}}

\begin{abstract}
\noindent Unclonable cryptography utilizes the principles of quantum mechanics to address cryptographic tasks that are impossible to achieve classically. We introduce a novel unclonable primitive in the context of secret sharing, called unclonable secret sharing (USS). In a USS scheme, there are $n$ shareholders, each holding a share of a classical secret represented as a quantum state. They can recover the secret once all parties (or at least $t$ parties) come together with their shares. Importantly, it should be infeasible to copy their own shares and send the copies to two non-communicating parties, enabling both of them to recover the secret. 

Our work initiates a formal investigation into the realm of unclonable secret sharing, shedding light on its implications, constructions, and inherent limitations.
     \begin{itemize}
         \item {\bf Connections}: We explore the connections between USS and other quantum cryptographic primitives such as unclonable encryption and position verification, showing the difficulties to achieve USS in different scenarios. 
         
         \item {\bf Limited Entanglement}: In the case where the adversarial shareholders do not share any entanglement or limited entanglement, we demonstrate information-theoretic constructions for USS. 

         \item {\bf Large Entanglement}: If we allow the adversarial shareholders to have unbounded entanglement resources (and unbounded computation), we prove that unclonable secret sharing is impossible. On the other hand, in the quantum random oracle model where the adversary can only make a bounded polynomial number of queries, we show a construction secure even with unbounded entanglement.  
         
         Furthermore, even when these adversaries possess only a polynomial amount of entanglement resources, we establish that any unclonable secret sharing scheme with a reconstruction function implementable using Cliffords and logarithmically many {\sf T}-gates is also unattainable.  
     \end{itemize}

\qipeng{
Some directions: 
\begin{itemize}
    \item (Computational Assumptions) In the plain model, with unbounded shared entanglement
    \item Interpolate the impossibility, possibly in terms of the number of $T$ gates 
    \item Applications of unclonable secret sharing
\end{itemize}
}
\vipul{Here is a weaker notion of USS which actually suffices for our motivation in the introduction: one of the two reconstructing parties (say Bob) is actually honest and will use the honest reconstruction algorithm. The other one would use any arbitrary construction procedure. We require that if Bob recovers the original Secret, the secret recovered by Charlie must be unrelated to the original secret. So we might be able to obtain stronger constructions for this definition.}

\vipul{How about considering ``uncloneable computation"? Just consider a client storing two shares of her data on two servers: servers should be able to run 2PC on their shares and compute a function output. However even given their views in 2PC (and original shares), they shouldn't be able to clone their shares. Natural application is: Alice secret shares her data across two clouds. But rather than just retrieving the data at a later point, Alice is also interested in computing on the data (e.g., doing a keyword search). But Alice still doesn't want cloud providers to keep copies of her data. Overall, seems like unless we are able to compute on the uncloneable shares, applications of uncloneable secret sharing might be somewhat limited.}

\end{abstract}

\newpage
\setcounter{tocdepth}{2}

\ifllncs\else
\tableofcontents
\newpage
\fi

\section{Introduction}

Alice is looking for storage for her sensitive data. She decides to hire multiple independent cloud providers and secret shares her data across them. Later on, Alice retrieves these shares and reconstructs the data. Everything went as planned. However: what if the cloud providers keep a copy and sell shares of her data to her competitor, Bob? How can Alice make sure that once she retrieves her data, no one else can?

This is clearly impossible in the classical setting. The cloud providers can always keep a copy of the share locally and later, if Bob comes along, sell that copy to Bob. Nonetheless, this problem has been recently studied in the classical setting by a recent work of Goyal, Song, and  Srinivasan~\cite{goyal2021traceable} who introduced the notion of traceable secret sharing (TSS). In TSS, if (a subset of) the cloud providers sell their shares to Bob, they cannot avoid leaving a cryptographic proof of fraud with Bob. Moreover, this cryptographic proof could not have been generated by Alice. Hence, (assuming Bob cooperates with Alice), Alice can sue the cloud providers in court and recover damages. Thus, TSS only acts as a deterrent and indeed, cannot stop the cloud providers from copying the secret.


However, in the quantum setting, the existence of no cloning theorem offers the tantalizing possibility that perhaps one may be able to build an ``unclonable secret sharing'' (USS) scheme. Very informally, the most basic version of a USS can be described as follows:

\begin{itemize}
    \item Alice (the dealer) has a classical secret $m \in \{0,1\}^*$. She hires $n$ cloud providers   ${\cal P}_1,\ldots,{\cal P}_n$. 
    \item Alice computes shares $(\rho_1, \cdots, \rho_n)$, which is an $n$-partite state, from $m$ and sends the share $\rho_i$ to the party ${\cal P}_i$ (note that Alice does not need to store any information like a cryptography key on her own). 
    \item Given $(\rho_1, \cdots, \rho_n)$, it is easy to recover $m$. But given any strict subset of the shares, no information about $m$ can be deduced (i.e., it is an $n$-out-of-$n$ secret sharing scheme).
    \item The most important is the unclonability. For every $i \in [n]$, the party ${\cal P}_i$ computes a bipartite state $\sigma_{\regX_i \regY_i}$. It sends the register ${\regX_i}$ to Bob and $\regY_i$ to Charlie. Assuming that the message $m$ was randomly chosen to be either $m_0$ or $m_1$ (where $(m_0,m_1)$ is chosen adversarially), the probability that both Bob and Charlie can guess the correct message must be upper bounded by a quantity negligibly close to $\frac{1}{2}$.
\end{itemize}

In other words, the parties ${\cal P}_1,\ldots,{\cal P}_n$ must be unable to locally clone their shares such that both sets of shares allow for reconstruction. Indeed, as we mentioned, this is the most basic version of USS. Even this basic setting has a practical significance: the servers which store Alice's shares may not intentionally communicate her shares with each other, because they belong to companies with conflict of interest; but a malicious Bob may still buy a copy of Alice's share from each of them.

One can consider more general settings where, e.g., we are interested in threshold (i.e., $t$-out-of-$n$) USS or, where a subset of the $n$ parties might collude in attempting to clone their shares. One can also consider the setting where the parties ${\cal P}_1,\ldots,{\cal P}_n$ share some entanglement (allowing them to use quantum teleportation). 

Unclonable cryptography leverages the power of quantum information and empowers one to achieve primitives which are clearly impossible in classical cryptography. While a lot of efforts have been made towards various unclonable cryptographic primitives including but not limited to quantum money~\cite{BB84,AC12,Zha17,Shm22,liu2023another}, copy-protection~\cite{Aar09,CLLZ21,AL20}, tokenized signatures~\cite{BDS16,CLLZ21,Shm22} and unclonable encryption (UE)~\cite{Got02,BL20,AK21,AKLLZ22,AKL23}, the question of unclonable secret sharing had not been studied prior to our work. Secret sharing is one of the most fundamental primitives in cryptography and as such, we believe that studying unclonable secret sharing is an important step towards laying the foundation of unclonable cryptography. Our contribution lies in initiating a systematic study of USS. 

\paragraph{Connection to Unclonable Encryption.}
The classical counterparts of unclonable encryption and (2-out-of-2) unclonable secret sharing are very similar. For instance, both one-time pad encryption and 2-out-of-2 secret sharing rely on the same ideas in the classical setting. One may wonder if UE and USS share similar a relation. UE resembles standard encryption with one additional property: now ciphertext is unclonable, meaning no one can duplicate a ciphertext into two parts such that both parts can be used separately to recover the original plaintext. At first glance, it might seem like UE directly implies a 2-out-of-2 USS. To secret share $m$, the dealer (Alice) would generate a secret key $sk$, and compute ciphertext $\rho_{\ct}$, which encrypts the classical message $m$. One of the shares will be $\rho_{\ct}$ while the other will be $sk$. Since $\rho_\ct$ is unclonable, this may prevent two successful reconstructions of the original message. 

However, the above intuition does not work if the two parties in (2-out-of-2) USS share entanglement. In UE, the ciphertext $\rho_{\ct}$ is a split into two components and sent to Alice and Bob. Later on, the secret key $sk$ is sent (without any modification) to both Alice and Bob. However, in USS, the secret key $sk$ corresponds to the second share and might also be split into two register such that one is sent to Alice and the other to Bob. This split could be done using a quantum register which is entangled with the quantum register used to split the cipher text $\rho_{\ct}$. It is unclear if such an attack can be reduced to the UE setting, where there is no analog of such an entangled register. In fact, we show the opposite. We show that in some settings, USS implies UE, thus showing that USS could be a stronger primitive.

\paragraph{Connection to Instantaneous Non-Local Computation.} It turns out that the positive results on instantaneous non-local computation imply negative results on USS in specific settings. The problem of instantaneous non-local computation~\cite{Vai03,beigi2011simplified,speelman15,ishizaka2008asymptotic,GC20} is the following: Dave and Eve would like to compute a unitary $U$ on a state $\rho_{\mathbf{X} \mathbf{Y}}$, where Dave has the register $\mathbf{X}$ and Eve has the register $\mathbf{Y}$. They need to do so by just exchanging one message simultaneously with each other. Non-local computation has connections to the theory of quantum gravity, as demonstrated in some recent works~\cite{may2019quantum, May23}. Suppose there is a unitary $U$ for which non-local computation is possible then this rules out a certain class of unclonable secret sharing schemes. Specifically, it disallows certain reconstruction procedures that are functionally equivalent to $U$. In more detail, consider a USS scheme that is defined as follows: on input a message $m$, it produces shares on two registers ${\bf X}$ and ${\bf Y}$. The reconstruction procedure\footnote{In general, a reconstruction procedure need not output a copy of the secret twice but using CNOT gates, we can easily transform any reconstruction procedure into one that outputs two copies of the secret.} takes as input the shares and outputs $m$ in both registers ${\bf X}$ and ${\bf Y}$. Any non-local computation protocol for such a reconstruction procedure would violate the security of the USS scheme. Investigating both positive and negative results of USS schemes could shed more light on the feasibility of non-local computation. In this work, we adapt and generalize techniques used in the literature on non-local computation to obtain impossibility results for USS. \\

\noindent USS also has connections to position verification, a well-studied notion in quantum cryptography that has connections to problems in fundamental physics. We discuss this in the next section. 
 
\subsection{Our Results}

In this work, our primary emphasis will be on $n$-out-of-$n$ unclonable secret sharing schemes as
even though they are the simplest, they give rise to numerous intriguing questions. Our results are twofold, as below. 

\subsubsection{Results on Information-Theoretic USS}

\begin{figure}[!hbt]
\centering
\begin{tikzpicture}
  \node[anchor=north east] at (10,2) {Information-Theoretic};
  \draw (6,1.5) rectangle (10,2);

  \node (uss1) at (0,0) {$\uss_{1}$};
  \node (ue) at (3,0) {$\UE$};
  \node (uss2) at (6, 0) {$\uss_2$};
  \node (dots) at (7.5, 0) {$\cdots$};
  \node[below] at (dots.south) {trivial};
  \node (ussn) at (10, 0) {$\uss_{\omega(\log \lambda)}$};
  \node[below] at (ussn.south) {\begin{tabular}{c} construction (d) \\ \Cref{sec:IT-scheme-many-components} \end{tabular}};
  \draw[->] (uss1) -- node[midway, below] { \begin{tabular}{c} \Cref{sec:USS_implies_UE} \\ (a) \end{tabular}} (ue);
  \draw[->] (ue) -- node[midway, below] { \begin{tabular}{c}\Cref{sec:UE_implies_USS2} \\ (c) \end{tabular} }  (uss2);
  \draw (uss2) -- (dots);
  \draw[->] (dots) -- (ussn);
  \draw[->] (ue) to [out=120,in=60,looseness=0.6] node[midway] {$\times$} node[midway, above] {\begin{tabular}{c} \Cref{sec:general_impossibility} \\ (b) \end{tabular}} (uss1);
  \draw[->] (ussn) to [out=135,in=45,looseness=0.3] node[midway] {\begin{tabular}{c} (e) \\ ?  \end{tabular}} (ue);
\end{tikzpicture}
\caption{Relations between USS and UE in the information-theoretic regime.}
\label{fig:IT-relations}
\end{figure}
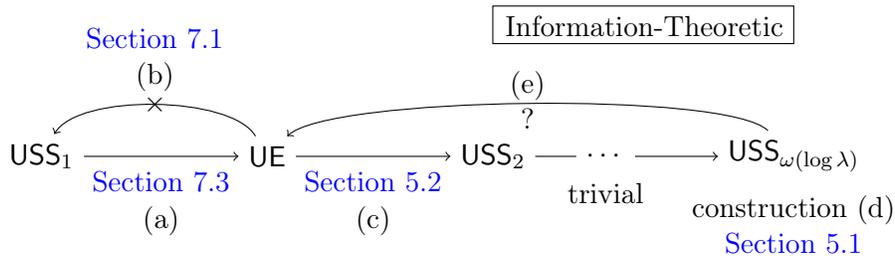

We first examine the connections between USS and UE and constructions of UE in the information-theoretic regime. The first part of our results can be summarized by \Cref{fig:IT-relations}. In the figure, $\uss_1$ stands for information-theoretic USS, secure against adversarial parties sharing \emph{unbounded} amount of entanglement; we will explain why we call it  $\uss_1$ later on. 
We first show that, even if we restrict adversaries in $\uss_1$ to have a polynomial amount of entanglement, it implies UE. \vipul{USS1 is defined is having unbounded amount of entanglement while the theorem below says bounded. We should probably not define USS1 as having unbounded.} \qipeng{added an explaination.}
\begin{theorem}[direction (a) in \Cref{fig:IT-relations}, \Cref{{sec:USS_implies_UE}}]
    Information-theoretic USS that is secure against adversarial parties $\cal{P}$ sharing \emph{polynomial} amount of entanglement implies UE. 
\end{theorem}

This leads us to ponder whether $\uss_1$ and UE share equivalence, like their classical counterparts do. Perhaps surprisingly, we show that this connection is unlike to hold. We prove that $\uss_1$ does not exist in the information-theoretic setting. Since there is no obvious evidence to refute UE in the IT setting and many candidates were proposed toward information-theoretic UE, our impossibility stands in sharp contrast to UE. 
\begin{theorem}[direction (b) in \Cref{fig:IT-relations}, \Cref{sec:general_impossibility}]
Information-theoretic USS that is secure against adversarial parties $\cal{P}$ sharing \emph{unbounded} amount of entanglement with each other, does not exist. 
\end{theorem}

Facing the above impossibility, it seems like USS in the IT regime comes to a dead end. To overcome the infeasibility result, we investigate USS against adversarial parties with specific entanglement configurations. We consider the case where every pair of ${\cal P}_i$ and ${\cal P}_j$ either shares unbounded entanglement or shares no entanglement. In this case, we can define an entanglement graph, of which an edge $(i, j)$ corresponds to entanglement between ${\cal P}_i$ and ${\cal P}_j$. Then, we propose the natural generalization and define $\uss_d$ for any $d > 1$: 
\begin{description}
    \item $\uss_d$: Information-theoretic USS, secure against adversarial parties sharing entanglement whose entanglement graph has at least $d$ connected components. 
\end{description}

The above definition captures the case that there are $d$ groups of parties; there is unlimited entanglement between parties in the same group and no entanglement between parties in different groups. This notation is not only for overcoming the barrier, but also has practical interest: parties from different groups are geographically separated or have conflict of interest, maintaining entanglement between them is either too expensive or impossible. Note that the characterization of entanglement is only for adversarial parties, whereas honest execution of the scheme does not need any pre-shared entanglement. We also like to note that aforementioned $\uss_1$ is also captured by the above definition when $d=1$. 

It is easy to see that the existence of $\uss_d$ implies $\uss_{d+1}$ for any $d \geq 1$, as having less entanglement makes attacking more difficult. However, since $\uss_1$ is impossible, can we construct $\uss_d$ for some $d$? We complete the picture of USS and UE by presenting the following two theorems. 
\begin{theorem}[direction (c) in \Cref{fig:IT-relations}, \Cref{sec:UE_implies_USS2}]
    $\UE$ implies $\uss_2$ in the information-theoretic setting. As a corollary, it implies $\uss_d$ for any $d > 1$ in the IT setting. 
\end{theorem}
\begin{theorem}[construction (d) in \Cref{fig:IT-relations}, \Cref{sec:IT-scheme-many-components}] \label{thm:uss_it_construction_intro}
    $\uss_d$ exists for every $d = \omega(\log \lambda)$ in the information-theoretic setting, where $\lambda$ is the security parameter. 
\end{theorem}
\noindent 
Along with \Cref{thm:uss_it_construction_intro}, we proved a special XOR lemma of the well-known monogamy-of-entanglement property for BB84 states~\cite{BB84,TFKW13}, when the splitting adversary is limited to tensor strategies. More precisely, we only consider cloning strategies that apply channels on each individual qubit, but never jointly on two or more qubits.  
Given a BB84 state, let $p(n)$ be the probability of the optimal tensor cloning strategy, that later two non-communicating parties recover the parity simultaneously. $p(1) = 1/2+1/{2\sqrt{2}}$ was proved in~\cite{TFKW13}. In this work, we show that $p(n) = 1/2 + \exp(-\Omega(n))$, which demonstrates a XOR hardness amplification for tensor strategies. We believe the proof of the theorem will be of independent interest, as a more general version of the theorem (that applies to any cloning strategies) will imply UE in the IT setting, resolving an open question on unclonable encryption since~\cite{BL20}.  \vnote{making this paragraph more precise would be helpful} \qipeng{how about this?}

\medskip

These two theorems establish a clear distinction between $\uss_1$ and $\uss_d$ for all $d$ greater than 1. Furthermore, the latter theorem illustrates that as the value of $d$ becomes sufficiently large, it becomes feasible to achieve $\uss_d$ within the IT setting. Consequently, it implies that, at the very least, certain objectives outlined in \Cref{fig:IT-relations} can be constructed.

\medskip

Lastly, as the final arrow in \Cref{fig:IT-relations}, does $\uss_2$ or $\uss_{\omega(\log \lambda)}$ implies UE? 
\begin{remark}[direction (e) in \Cref{fig:IT-relations}]
    We do not have an answer yet. Nonetheless, we assert that either $USS_d$ does not imply $\UE$, or establishing this implication is as challenging as constructing UE. The latter assertion arises from our existing knowledge of $\uss_{\omega(\log \lambda)}$ --- demonstrating such an implication should, in turn, furnish us with a means to construct UE within the IT framework.
\end{remark}

\subsubsection{Results on Computational USS}

In this computational regime, adversarial parties are computationally bounded; this in turn implies that the amount of pre-shared entanglement is also computationally bounded. Unlike the comprehensive picture presented in \Cref{fig:IT-relations}, our understanding here is more intricate. Specifically, as demonstrated in \Cref{fig:comp-relations}, the feasibility or infeasibility hinges on factors such as the computational complexity of USS schemes and the actual quantity of shared entanglement among malicious parties.

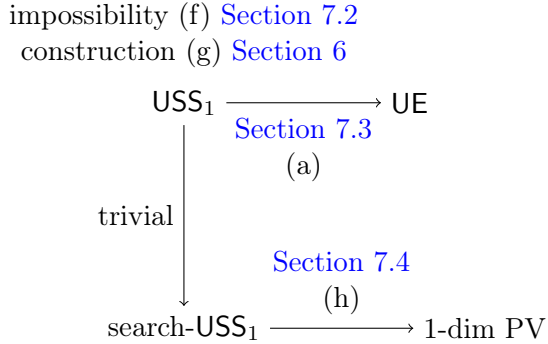
\begin{figure}[!hbt]
\centering
\begin{tikzpicture}
  \node[anchor=north east] at (10,2) {Computational};
  \draw (7.25,1.5) rectangle (10,2);

  \node (uss1) at (0,0) {$\uss_{1}$};
  \node (ue) at (3,0) {$\UE$};
  \node (uss2s) at (0, -3) {search-$\uss_1$};
  \node (pv) at (4, -3) {$1$-dim PV};
  
  \draw[->] (uss1) -- node[midway, below] { \begin{tabular}{c} \Cref{sec:USS_implies_UE} \\ (a) \end{tabular}} (ue);
  \draw[->] (uss1) -- node[midway, left] { trivial } (uss2s);
  \draw[->] (uss2s) -- node[midway, above] { \begin{tabular}{c}\Cref{sec:uss_imply_PV} \\ (h) \end{tabular} }  (pv);

  \node[above] (construction) at (uss1.north) {\begin{tabular}{c} impossibility (f)  \Cref{sec:impossibility_lowT} \\  construction (g)  \Cref{sec:construction_qrom}  \end{tabular}};
\end{tikzpicture}
\caption{Relations between USS and UE in the computational regime.}
\label{fig:comp-relations}
\end{figure}

Similar to the IT setting, the implication of $\uss_1$ and UE still works (direction (a) in \Cref{fig:comp-relations}). What is new here is that we present one impossibility result and one infeasibility result on $\uss_1$. 
\begin{theorem}[Informal, impossibility (f) in \Cref{fig:comp-relations}, \Cref{sec:impossibility_lowT}]
    USS whose reconstruction function has only $d$ {\sf T} gates, can be attacked with adversarial parties sharing $O(2^d)$ qubits of pre-shared entanglement. 
\end{theorem}
\noindent Therefore, when the reconstruction has low {\sf T} complexity, say $d = \log \lambda$, then such $\uss$ does not exist even in the computational regime. Next, we present a construction, in sharp contrast to the impossibility above. Quantum random oracle~\cite{boneh2011random}, models the perfect (and unrealizable) cryptographic hash function. As it should behave as a truly random function, it can not have a small number of {\sf T} gates. 
\begin{theorem}
\label{thm:construction_qrom}
[construction (g) in \Cref{fig:comp-relations}, \Cref{sec:construction_qrom}]
    USS that is secure against query-efficient adversarial parties sharing an arbitrary amount of pre-shared entanglement\footnote{The adversary is polynomially bounded in queries but not in the pre-shared entanglment.}, exists in the quantum random oracle model (QROM). 
\end{theorem}

As quantum random oracle is not realizable in general, we wonder whether $\uss_1$ can be constructed in the plain model. To the end, we show that $\uss_1$ implies a cryptographic primitive called $1$-dimensional position verification that is secure against parties sharing any polynomial amount of entanglement. Position verification represents an actively explored research area. Despite all the ongoing efforts, the development of a construction for position verification within the standard model remains elusive. This underscores the formidable challenge of devising $\uss_1$, when relying on computational assumptions.
\begin{theorem}[direction (h) in \Cref{fig:comp-relations}, \Cref{sec:uss_imply_PV}]
    USS that is secure against adversarial parties having pre-shared entanglement, implies $1$-dimensional position verification that is secure against parties sharing the same amount of pre-shared entanglement.
\end{theorem}

\subsection{Other Related Works}
\paragraph{On Secret Sharing of Quantum States}
Our work focuses on secret-sharing classical secrets by encoding them into a quantum state to achieve unclonability. One may be curious about the relationship of our new primitive to the existing studies on secret-sharing schemes where the secret messages are \emph{quantum states} to begin with.

In short, all the existing quantum secret sharing schemes fall short of satisfying one crucial property in our model: the requirement of \emph{no or low entanglement} for honest parties. Their unclonability also remains elusive, as they require much more complicated structures on quantum states than ours. We provide a detailed discussion below and will carefully incorporate all the discussions into the subsequent version. 

In the paper, we consider a model where malicious parties can share some amount of entanglement before attacking the protocol. As illustrated in \Cref{fig:IT-relations} and \Cref{fig:comp-relations}, the amount of entanglement (or more precisely, the entanglement graph) plays an important role in both the construction and barriers of such schemes. Therefore, we do not want the entanglement used in honest shares to scale to the same order or surpass what adversaries can access. Our constructions (\Cref{thm:uss_it_construction_intro} and \Cref{thm:construction_qrom}) are based on unentangled quantum shares of single qubits, thus no entanglement required. 

\cite{quantumsecresharing1999} first proposed the idea of using quantum states to secret-share a classical bit. Their idea is to use $n$-qubit GHZ states for an $n$-out-of-$n$ secret share scheme. However, an $n$-qubit GHZ state requires entanglement across $n$ quantum registers, which enforces shareholders to maintain entanglement with each other. A subsequent proposal in \cite{PhysRevA.59.162} followed a similar path but also required a large amount of entanglement. The idea of using quantum state to secret share classical secrets was also discussed by Gottesman \cite{gottesman2000theory}, but they mostly focused on the lower bounds of general schemes (potentially requiring entanglement): for example, how many qubits are required to secret-share one classical bit. 

There is another line of works on secret-sharing quantum secrets, including \cite{howtosharequantumsecret1999cleve},\cite{smith2000quantum}  and most recently \cite{ccakan2023computational} by {\c{C}}akan et al. Since the goal is to secret-share a quantum state, entanglement is also necessary in these protocols.

\section{Technical Overview}

In this section, unless otherwise specified, we focus on $2$-out-of-$2$ USS, with $\share$ and $\reconstruct$. $\share$ takes as input a message $m$ and outputs two shares $\rho_0, \rho_1$; whereas $\reconstruct$ takes two quantum shares and outputs a string. We assume $\rho_0, \rho_1$ are unentangled. When we consider impossibility results, all arguments mentioned in this overview carry in the same way to the general cases; for constructions, we only require unentangled shares. 

\subsection{\texorpdfstring{$\uss_1$ implies $\UE$, $\UE$ implies $\uss_2$}{USS implies UE, UE implies USS2}}

We first examine two directions (directions (a) and (c) in \Cref{fig:IT-relations,fig:comp-relations}); that is, how $\uss_1$ implies $\UE$ and how $\UE$ implies $\uss_2$. These two directions work in both IT and computational setting. We briefly recall the definition of UE: it is a secret key encryption scheme with the additional property: there is no way to split a quantum ciphertext into two parts, both combining with the classical secret key can recover the original plaintext (with probability at least $1/2$ plus negligible). 

\paragraph{$\uss_1$ implies $\UE$, \Cref{sec:USS_implies_UE}.}  Given a $2$-out-of-$2$ USS, we now design a UE:
\begin{description}
    \item $\UE.\enc(k, m)$ takes as input a secret key $k$ and a message, 
    \begin{enumerate}
        \item it first produces two shares $(\rho_1, \rho_2) \gets \uss.{\share}(m)$,
        \item it parses $k = (a, b)$ and let the unclonable ciphertext be $\ct = (\rho_1, X^a Z^b \rho_2 Z^b X^a)$. In other words, it sends out $\rho_1$ in clear, while having $\rho_2$ one-time padded by the key $k$. 
    \end{enumerate}
\end{description}
Decryption is straightforward, by unpadding $X^a Z^b \rho_2 Z^b X^a$ and applying $\reconstruct$ to $(\rho_1, \rho_2)$. Correctness and semantic security follows easily. Its unclonability can be based on the unclonability of $\uss_1$; indeed, the scheme corresponds to a special strategy of malicious ${\cal P}_1$ and ${\cal P}_2$. Suppose there exists an adversary $(\alice,\bob,\charlie)$ that violates the above scheme, there exists $({\cal P}_1,{\cal P}_2,\bob,\charlie)$ that violates the security of $\uss_1$.

\begin{description}
    \item ${\cal P}_1$ and ${\cal P}_2$ share EPR pairs. ${\cal P}_2$ uses the EPR pairs to teleport $\rho_2$ to ${\cal P}_1$, with ${\cal P}_2$ having random $(a, b)$ and ${\cal P}_1$ obtaining $(\rho_1, X^a Z^b \rho_2 Z^b X^a)$. As ${\cal P}_2$ only has classical information, it sends $(a, b)$ to both $\bob$ and $\charlie$, while ${\cal P}_1$ applies $\alice$ on  $(\rho_1, X^a Z^b \rho_2 Z^b X^a)$ and shares the bipartite state with both $\bob$ and $\charlie$.
\end{description}

\noindent It is not hard to see that the above attacking strategy for $\uss_1$ exactly corresponds to an attack in the $\UE$ we proposed above: ${\cal P}_1$ tries to split a ciphertext while ${\cal P}_2$ simply forwards the secret key $k = (a, b)$. Therefore, we can base the unclonability of the $\UE$ on that of $\USS_1$, which completes the first direction.  

\paragraph{$\UE$ implies $\uss_2$, \Cref{sec:UE_implies_USS2}.} Recall that $2$-out-of-$2$ $\uss_2$ describes adversarial parties who do not share any entanglement. We can simply set up our $\uss_2$ scheme as follows, using $\UE$:
\begin{description}
    \item $\share(m)$ takes as input a message $m$, it samples a key $k$ for $\UE$, and let $\ket{\ct}$ be the unclonable ciphertext of $m$ under $k$; the procedure $\share$ outputs the first share as $\rho_1 = k$, and the second share as $\rho_2 = \ket{\ct}$.
\end{description}
As there is no entanglement between ${\cal P}_1$ and ${\cal P}_2$, ${\cal P}_1$ with $\rho_1 = k$ forwards the classical information to both Alice and Bob. In the meantime, ${\cal P}_2$ employs her cloning strategy, which remains entirely independent of the key $k$. Consequently, the unclonability of out $\uss_2$ aligns with that of $\UE$.

When we generalize the conclusion to $n$-out-of-$n$ $\uss_2$, we first secret share the targeted message $m$ into $n$ shares. For any two adjacent parties ${\cal P}_i$, ${\cal P}_{i+1}$ and the $i$-th share, the first part receives the key and the second one gets the unclonable ciphertext. As long as all the malicious parties form at least two connected components (as defined in $\uss_2$), there must be two adjacent parties who do not have entanglement. Thus, we can incur the same logic to prove its unclonability, basing on the unclonability of $\UE$.

\subsection{\texorpdfstring{Construction of $\uss_{\omega(\log \lambda)}$}{Construction of USS with omega(log lambda) connected components}}
\label{sec:tech_overview_USS_disconnected}
For simplicity, we focus on an $n$-out-of-$n$ USS, where $n = \omega(\log \lambda)$ and no entanglement is shared between any malicious parties, which is a special case of a general $n$-out-of-$n$ $\uss_{\omega(\log \lambda)}$, for a larger $n \gg \omega(\log \lambda)$. Our construction is based on the BB84 states.
Our scheme first classically secret-shares $m$ into $(n-1)$ shares and encodes each classical share into a single-qubit BB84 state. One party will receive the basis information $\theta$ which contains $(n-1)$ basis; every other party will receive a BB84 state for the $i$-th classical share. 
\begin{description}
    \item $\share(m)$: it takes as input a secret $m \in \{0,1\}$,
    \begin{itemize}
        \item it samples $m_1, \cdots, m_{n-1}$ conditioned on their parity equals to $m$;
        \item it samples $\theta \in \{0,1\}^{n-1}$;
        \item let the first $(n-1)$ shares be $\rho_i = H^{\theta_i} \ket{m_i}\bra{m_i} H^{\theta_i}$ and the last share $\rho_n = \ketbra \theta \theta$.
    \end{itemize}
\end{description}
Reconstruction of shares is straightforward. After receiving all shares, one uses the basis information $\theta$ to recover all the classical shares $m_i$; $m$ then is clearly determined by these $m_i$. 

To reason about the unclonability of our protocol, we first recall a theorem on BB84 states, initially proposed by Tomamichel, Fehr, Kaniewski and Wehner~\cite{TFKW13} and later adapted in constructing unclonable encryption by Broadbent and Lord~\cite{BL20}. We start by considering a cloning game of single-qubit BB84 states. 
\begin{enumerate}
    \item $\alice$ receives $H^\theta \ketbra{x}{x} H^\theta$ for uniformly random $x, \theta \in \{0,1\}$, it applies a channel and produces $\sigma_{\mathbf{B} \mathbf{C}}$. Bob and Charlie receive their registers accordingly. 
    \item Bob $\bob$ and Charlie $\charlie$ apply their POVMs and try to recover $x$; they win if and only if both guess $x$ correctly. 
\end{enumerate}
\begin{lemma}[Corollary 2 when $n = 1$, \cite{BL20}]\label{lem:single_qubit_bb84}
    No (unbounded) quantum $(\alice, \bob, \charlie)$ wins the above game with probability more than 0.855.
\end{lemma}
Tomamichel, Fehr, Kaniewski and Wehner~\cite{TFKW13} and Broadbent and Lord~\cite{BL20} studied parallel repetitions of the above cloning game\footnote{Indeed, \cite{TFKW13} proved a stronger statement on a different game, which ultimately implied the parallel repetition theorem, shown by~\cite{BL20}.}. In the parallel repetition, $n$ random and independent BB84 states are generated, which encode an $n$-bit string $x$. The goal of cloning algorithms is to guess the $n$-bit string $x$ simultaneously. They showed that the cloning game follows parallel repetition, meaning that the optimal winning probability in an $n$-fold parallel repetition game is at most $(0.855)^n$. 

Our proposed scheme also prepares these BB84 states in parallel, but hides the secret $m$ as the XOR of the longer secret. Indeed, the XOR repetition of the BB84 cloning game has been a folklore and was considered as a candidate for UE. More specifically, it is conjectured that the following game can not be won by any algorithm with probability more than $1/2 + \exp(-\Omega(n))$: 
\paragraph{XOR repetition of BB84 cloning games.}
\begin{enumerate}
    \item $\alice$ receives $H^\theta \ketbra{x}{x} H^\theta$ for uniformly random $x, \theta \in \{0,1\}^n$, it applies a channel and produces $\sigma_{\mathbf{B} \mathbf{C}}$. Bob and Charlie receive their register accordingly. 
    \item Bob $\bob$ and Charlie $\charlie$ apply their POVMs and try to recover ${\sf parity}(x)$; they win if and only if both guess correctly. 
\end{enumerate}
Although there is no evidence to disprove the bound for the XOR repetition so far, the validity of the bound still remains unknown. In this work, we prove this bound, when $\alice$ is restricted to a collection of strategies. It applies ${\cal C}_i$ on the $i$-th qubit of the BB84 state and get $\sigma^{(i)}_{\mathbf{BC}}$; the final state $\sigma_{\mathbf{BC}} = \bigotimes_i \sigma^{(i)}_{\mathbf{BC}}$. Note that the lemma does not put any constraint on the behaviors of $\bob$ or $\charlie$. 
\begin{lemma}[An XOR lemma for BB84 cloning games, \Cref{sec:IT-scheme-many-components}.]
\label{lem:IT-many-components-tech-overview}
    When $\alice$ only applies a tensor cloning strategy to prepare $\sigma_{\mathbf {BC}}$, the optimal success probability in the XOR repetition of BB84 games is $1/2 + \exp(-\Omega(n))$. 
\end{lemma}
\noindent Equipped with it, it is straightforward to show the unclonability of our protocol. 

\medskip

\paragraph{A proof for the XOR repetition.}
Finally, we give a brief recap on the proof for \Cref{lem:IT-many-components-tech-overview}. 

For any $\alice$'s tensor strategy with channels $\mathcal{C}_i$ applied on the $i$-th qubit of a BB84 state, we recall the notation $\sigma^{(i)}_{\mathcal{BC}}$. This is the state produced from the $i$-th qubit of the B884 state, when $\theta_i, x_i$ was sampled uniformly at random. 
Let $\sigma^{(i, 0)}_{\mathbf{B}}$ be the density matrix, describing the register that will be given to Bob, when $x_i = 0$. We can similarly define $\sigma^{(i, 1)}_{\mathbf{B}}$, $\sigma^{(i, 0)}_{\mathbf{C}}$ and $\sigma^{(i, 1)}_{\mathbf{C}}$. \Cref{lem:single_qubit_bb84} tells us that, there exists a constant $c > 0$, either 
\begin{align*}
    {\sf TD}(\sigma^{(i, 0)}_{\mathbf{B}}, \sigma^{(i,1)}_{\mathbf B}) < c  \quad\quad\quad \text{ or } \quad\quad\quad {\sf TD}(\sigma^{(i, 0)}_{\mathbf{C}}, \sigma^{(i,1)}_{\mathbf C}) < c.
\end{align*}
This indicates that for every $i$, either Bob or Charlie can not perfectly tell the value of $x_i$, regardless of the channel $\mathcal{C}_i$. Furthermore, as the BB84 state has $n$ qubits, w.l.o.g. we can assume that the above holds for Bob, for at least $n/2$ positions. 

In the XOR repetition, Bob eventually will receive $\sigma^{(i, m_i)}_{\mathbf{B}}$. We show that Bob can not tell whether the parity of all $m_i$ is odd or even. More precisely, we will show:
$$\td\left(\sum_{\substack{m_1,\ldots,m_{n-1}:\\ \oplus_i m_i = 0}} \frac{1}{2^{n-2}} \left( \bigotimes_i \sigma^{(i, m_i)}_{\mathbf B} \right),\ \sum_{\substack{m_1,\ldots,m_{n-1}:\\ \oplus_i m_i = 1}} \frac{1}{2^{n-2}} \left( \bigotimes_i \sigma^{(i, m_i)}_{\mathbf B} \right)  \right) \leq  c^{n/2}.$$
We connect the trace distance directly to the trace distance of \emph{each pair of states} ${\sf TD}(\sigma^{(i, 0)}_{\mathbf{B}}, \sigma^{(i,1)}_{\mathbf B})$ and demonstrate \emph{an equality} (see \Cref{sec:IT-scheme-many-components}):
$$\td\left(\sum_{\substack{m_1,\ldots,m_{n-1}:\\ \oplus_i m_i = 0}} \frac{1}{2^{n-2}} \left( \bigotimes_i \sigma^{(i, m_i)}_{\mathbf B} \right),\ \sum_{\substack{m_1,\ldots,m_{n-1}:\\ \oplus_i m_i = 1}} \frac{1}{2^{n-2}} \left( \bigotimes_i \sigma^{(i, m_i)}_{\mathbf B} \right)  \right) = \prod_i {\sf TD}\left(\sigma^{(i, 0)}_{\mathbf{B}}, \sigma^{(i,1)}_{\mathbf B}\right).$$
Since every trace distance is bounded by $1$ and there are 
at least $n/2$ terms in the product smaller than $c$, we conclude the result.


\subsection{\texorpdfstring{Impossibility of $\uss_1$}{Impossibility of USS1}}
Since $\uss_1$ implies $\UE$, it is natural to consider building $\UE$ from $\uss_1$. Constructing $\UE$ in the basic model remained unresolved since~\cite{BL20}. Perhaps the connections in the last section provide a new avenue for constructing $\UE$. In this section, we present two impossibility results (referred to as (b) in \Cref{fig:IT-relations} and (f) in \Cref{fig:comp-relations}) that highlight challenges associated with $\uss_1$.

\paragraph{Information-theoretic $\uss_1$ does not exist, \Cref{sec:general_impossibility}.} 
We begin by examining the case of $2$-out-of-$2$ $\uss_1$ with unentangled shares, and our impossibility result extends to the general case. Let us consider two malicious parties, $\cP_1$ and $\cP_2$, who share an unlimited amount of entanglement. $\cP_2$ receives the initial share, $\rho_2$, and teleports it to $\cP_1$. This action leaves $\cP_2$ with a random one-time pad key, denoted as $(a, b)$ while $\cP_1$ now possesses $(\rho_1, X^a Z^b \rho_2 Z^b X^a)$. 
Now, $\cP_1$ aims to jointly apply the reconstruction procedure to $(\rho_1, \rho_2)$, but there's a problem: $\cP_1$ lacks all the necessary information, especially the one-time padded key. To address this challenge, we recall the concept of port-based teleportation~\cite{ishizaka2008asymptotic,beigi2011simplified} to help $\cP_1$.

Port-based teleportation allows one party to teleport a $d$-qubit quantum state to another party, while leaving the state in plain. This is certainly impossible without paying any cost, as it contradicts with special relativity. Two parties need to pre-share about $O(d 2^d)$ EPR pairs, divided into $O(2^d)$ blocks of $d$ qubits. After the port-based teleportation, the teleported state will be randomly dropped into one of the blocks of $\cP_2$, while only $\cP_1$ has the classical information about which block consists of the original state. 

Equipped with port-based teleportation, $\cP_1$ teleports $(\rho_1, X^a Z^b \rho_2 Z^b X^a)$ to $\cP_2$; it has the classical information ${\sf ind}$ specifying the location of the teleported state. $\cP_2$ then runs $\reconstruct \circ (I \otimes Z^b X^a)$ on every possible block among the pre-shared entanglement, yielding $O(2^d)$ different values; even though most of the execution is useless, the ${\sf ind}$-th block will store the correct (classical) answer. Finally, both $\cP_1$ and $\cP_2$ sends all their classical information to Alice and Bob; each of them can independently determine the message. This clearly violates the unclonability of $\uss_1$. Thus, for any 
$2$-out-of-$2$ $\uss_1$ whose shares are of length $d$, there is an attacking strategy that takes time and entanglement of order $\tilde{O}(d 2^d)$ and completely breaks its unclonability.

We refer readers to \Cref{sec:general_impossibility} for the proof of a general theorem statement. 

\paragraph{Impossibility of computationally secure $\uss_1$, with low-T $\reconstruct$, \Cref{sec:impossibility_lowT}.}
We now focus on the case when the reconstruction circuit can be implemented by Clifford gates and logarithmically many $\sf T$ gates. We would like to mention that a similar result has already been shown in \cite{speelman15} in the context of instantaneous non-local computation; we rediscovered the following simple attack for unclonable secret sharing. We also extend the attack to an $n$-party setting whereas \cite{speelman15} considers only 2 parties.

Denote $C$ to be the reconstruction circuit. That is, on input two shares of the form $\rho_1, \rho_2$, the output is the first bit of $C(\rho_1 \otimes \rho_2) C^\dagger = \ket m \bra m \otimes \tau$. 

We let $\cP_2$ teleport $\rho_2$ to $\cP_1$ and they try to compute $\reconstruct$ in a non-local manner. In the previous attack, this is done by leveraging an exponential amount of entanglement.
To avoid this and make the attack efficient, we hope that $\cP_1$ can homomorphically compute on the one-time padded data $(\rho_1, X^a Z^b \rho_2 Z^b X^a)$, without decrypting it. 

Suppose $C$ is a Clifford circuit. We use the fact that the Clifford group is a normalizer for the Pauli group (specifically, the $X^a Z^b$ operator). Let us assume each $\rho_1, \rho_2$ is of $\ell$ qubits.
In other words, for any $a, b \in \{0,1\}^\ell$ and Clifford circuit $C$, there exists a polynomial-time computable $a', b' \in \{0,1\}^{2 \ell}$ depending only on $a, b$ and $C$, such that
\begin{align*}
     C (\rho_1 \otimes X^{a} Z^{b} \rho_2 Z^{b} X^{a}) C^\dagger = X^{a'} Z^{b'} C (\rho_1 \otimes \rho_2) C^\dagger Z^{b'} X^{a'}. 
\end{align*}
Here $a', b'$ act as a bigger quantum one-time pad operated on $C (\rho_1 \otimes \rho_2) C^\dagger = \ket m \bra m \otimes \tau$. 

Now $\cP_1$ measures the first qubit in the computational basis, yielding $m \oplus a'_1$; whereas $\cP_2$ compute $a', b'$ (and most importantly, $a'_1$) from its classical information $a, b$. They send their knowledge to both Alice and Bob, who later simultaneously recover $m$. 

\medskip

Next, let us consider the more general case where $C$ consists of Clifford gates and $t$ number of $\sf T$ gates. The homomorphic evaluation of Clifford gates are as before. However, the homomorphic evaluation of $\sf T$ gates are handled differently.

Let us consider one single $\sf T$ gate that applies to the first qubit. We consider two identities, for any $x, z\in \{0,1\}$ and any single-qubit state $\ket{\psi}$
\begin{align*}
(i)\ T(X^x Z^z) \ket{\psi} &= (X^x Z^{x \oplus z} P^x) T \ket{\psi}, \\
(ii)\ P(X^x Z^z) \ket{\psi} &= (X^x Z^{x \oplus z}) P \ket{\psi}
\end{align*}
Suppose, the current state is of the form $X^x Z^z \ket{\psi}$ and we apply $P^{x}T$ to the state. We would like to show that the resulting state is $X^{a'}Z^{b'} T \ket{\psi}$ for some $a' \in \{0,1\},b' \in \{0,1\}$. We use the above identities:
$$ (P^{x} T) (X^x Z^z) \ket{\psi} \stackrel{\text{From }(i)}{=} P^{x} (X^x Z^{x \oplus z} P^x) T \ket{\psi} \stackrel{\text{From }(ii)}{=} X^x Z^{x \oplus z} P^{x \oplus x} T \ket{\psi}.$$
Note that $P^2 = P^0 = I$. 
Thus, if we can learn $x$ ahead, we can successfully homomorphic compute ${\sf T}$ on the encrypted data. However, in our case, $x$ corresponds to any bit in the one-time pad key $a$ of any stage. $\cP_1$ has no way to learn $x$. This is where the limitation of low-$\sf T$ gate comes from. Instead of knowing $x$ ahead, each time when a $\sf T$ homomorphic evaluation is needed, one simply guesses $x'$; as long as $x = x'$ (which happens with probability $1/2$), we succeed. Thus, $\cP_1$ only guesses all $x$'s (for each $\sf T$ gate) correctly with probability $2^{-t}$. If $t$ is logarithmic, our attack violates the security with inverse polynomial probability; therefore, it rules out computationally secure $\uss_1$ with a low-${\sf T}$ $\reconstruct$ procedure. 


\subsection{\texorpdfstring{Barries of $\uss_1$ (implication of PV)}{Barries of USS1}}
To further demonstrate the challenge of building $\USS$ against entangled adversaries, we show that 2-party $\USS_1$ implies a primitive called position verification.  Position verification (PV)  has remained a vexing problem since its inception~\cite{chandran2009position}.

 We briefly introduce the notion of position verification for the 1-dimensional setting: two verifiers on a line will send messages to a prover who claims to be located at a position between the two verifiers. By computing a function of the verifiers' messages and returning the answers to the verifiers in time, the prover ensures them of its location. 
However, two malicious provers may collude to impersonate such an honest verifier by standing at the two sides of the claimed position.  

We demonstrate that $2$-party $\USS_1$, even with the weaker search-based security, will imply PV: the two verifiers in the position verification protocol will generate secret shares $(\rho_0, \rho_1)$ of a random string $s$; then they will each send the messages $\rho_0$ and $\rho_1$ respectively to the prover; the prover needs to reconstruct $s$ and send $s$ to both verifiers in time. Any attack against PV can be viewed as a two-stage strategy---one can perfectly turn the first-stage strategy in PV into the shareholders'  strategy in $\USS$  and the second-stage strategy in PV into the recoverers' strategy  in $\USS$.

Despite many efforts, progress on PV in the computational setting against entangled adversaries has unfortunately been slow. 
We do not even know of any secure computational PV against adversaries with unbounded polynomial amount of entanglement in the plain model, nor any impossibility result. 
Moreover, some recent advancement in quantum gravity has unveiled some
connections between the security of position verification and problems in quantum gravity \cite{may2019quantum,May23}  .

Any progress of $\USS_1$ in the plain model will contribute towards resolving this long-standing open problem and unveil more implications. 

\section{Preliminaries}

\subsection{Notations}
We assume that the reader is familiar with the basic background from~\cite{nielsen2010quantum}. The Hilbert spaces we are interested in are $\mathbb{C}^d,$ for $d \in \mathbb{N}$. We denote the quantum registers with capital bold letters $\mathbf{R}$, $\mathbf{W}$, $\mathbf{X}$, ... \pnote{this is not consistently used everywhere. My suggestion would be to use bold letters since calligraphic is also used in other places.}\qipeng{I am working on it}. We abuse the notation and use registers in place of the Hilbert spaces they represent. The set of all linear mappings from $\mathbf{R}$ to $\mathbf{W}$ is denoted by $L(\mathbf{R}, \mathbf{W})$, and $L(\mathbf{R})$ denotes $L(\mathbf{R}, \mathbf{R})$. We denote unitaries with capital letters $C$, $E$, ... and the set of unitaries on register $\mathbf{R}$ with $U(\mathbf{R})$. We denote the identity operator on $\mathbf{R}$ with $\mathbb{I}_\mathbf{R}$; if the register $\mathbf{R}$ is clear from the context, we drop the subscript $\mathbf{R}$ from the notation $\mathbb{I}_\mathbf{R}$. We denote the set of all positive semi-definite linear mappings in $L(\mathbf{R}, \mathbf{R})$ with trace 1 (i.e., set of all valid quantum states) by $D(\mathbf{R})$. For a register $\mathbf{R}$ in a multi-qubit system, we denote $\overline{\mathbf{R}}$ to be a register consisting of all the qubits in the system not contained in $\mathbf{R}$. We denote $\Tr_{\mathbf{R}}(\rho)$ to be the state obtained by tracing out all the registers of $\rho$ except $\mathbf{R}$. A quantum channel $\Phi$ refers to a completely positive and trace-preserving (CPTP) map from a Hilbert  space $\mathcal{H}_1$ to a possibly different Hilbert  space $\mathcal{H}_2$.



\subsection{Unclonable Encryption}

Unclonable encryption was originally defined in \cite{BL20} and they considered two security notions, namely search and indistinguishability security, with the latter being stronger than the former. We consider below a mild strengthening of the indistinguishability security due to~\cite{AK21}. 

\vipul{In our disconnected component construction, we need unconditionally secure UE Without RO. I believe that is only known based on search based definition not captured in the definitions below. } \pnote{not sure I follow. For the implication, we need ind security which is what is defined below.}

\begin{definition}\label{def:ue}
An unclonable encryption scheme $\UE$ is a triple of efficient quantum algorithms $(\UE.\keygen,\UE.\enc,\UE.\dec)$ with the following procedures:
\begin{itemize}
    \item $\keygen(1^\lambda)$: On input a security parameter $1^\lambda$, returns a classical key $\sk$\footnote{In our construction, we require $\sk$ being a uniform random string. Such a $\UE$ scheme can be constructed in QROM \cite{AKLLZ22,AKL23}}.
    
    \item $\enc(\sk, m)$: It takes the key $\sk$ and the message $m$ for $m \in \{0,1\}^{\mathrm{poly}(\lambda)}$ as input and outputs a quantum ciphertext $\rho_{ct}$.  
    
    \item $\dec(\sk,\rho_{ct})$: It takes the key $\sk$ and the quantum ciphertext $\rho_{ct}$, it outputs a quantum state $\tau$.
\end{itemize}
\end{definition}

\paragraph{Correctness.} The following must hold for the encryption scheme. For every $\sk\leftarrow \keygen(1^{\lambda})$ and every message $m$, we must have $\Tr[\ket{m}\bra{m}\dec(\sk,\enc(\sk,\ket{m}\bra{m}))]\geq 1-\negl(\lambda)$.

\paragraph{Unclonability.} In the rest of the work, we focus on unclonable IND-CPA security. The regular IND-CPA security follows directly from its unclonable IND-CPA security. To define unclonable security, we introduce the following security game.

\begin{definition}[Unclonable IND-CPA game]
\label{def:ue_ind}
Let $\lambda\in \mathbb{N}^+$. Consider the following game against the adversary $(\As, \Bs, \Cs)$.
\begin{itemize}
    \item The adversary $\As$ generates $m_0,m_1\in\{0,1\}^{n(\lambda)}$ and sends $(m_0,m_1)$ to the challenger.
    \item The challenger randomly chooses a bit $b\in\{0,1\}$ and returns $\enc(\sk,m_b)$ to $\As$. $\As$ produces a quantum state $\rho_{\mathbf{B}\mathbf{C}}$ on registers $\mathbf{B}$ and $\mathbf{C}$, and sends the corresponding registers to $\mathcal{B}$ and $\mathcal{C}$.
    \item $\Bs$ and $\Cs$ receive the key $\sk$, and output bits $b_{\Bs}$ and $b_{\Cs}$ respectively.
\end{itemize}
The adversary wins if $b_\Bs=b_\Cs=b$. 
\end{definition}
\noindent We denote the success probability of the above game by $\advantage_{\As,\Bs,\Cs}(\lambda)$. We say that the scheme is information-theoretically (resp., computationally) secure if for all (resp., quantum polynomial-time) adversaries $(\mathcal{A},\mathcal{B},\mathcal{C})$,
\begin{align*}
\advantage_{\As,\Bs,\Cs}(\lambda)\leq 1/2+\negl(\lambda). 
\end{align*}

\subsection{Quantum Gate Sets}
\noindent We will work with the following quantum gate sets.  

\paragraph{Pauli Group.}
The single-qubit Pauli group $\mathcal{P}$ consists of the group generated by the following Pauli matrices: $$I = \begin{pmatrix} 1 & 0\\ 0 & 1 \end{pmatrix}\hspace{1cm} X = \begin{pmatrix} 0 & 1\\ 1 & 0 \end{pmatrix}\hspace{1cm} Y = \begin{pmatrix} 0 & i\\ -i & 0 \end{pmatrix}\hspace{1cm} Z = \begin{pmatrix} 1 & 0\\ 0 & -1 \end{pmatrix}$$ The $n$-qubit Pauli group $\mathcal{P}_n$ is the $n$-fold tensor product of $\mathcal{P}$.

\paragraph{Clifford Group.}
The $n$-qubit Clifford group is defined to be the set of unitaries $C$ such that
$$C\mathcal{P}_nC^{\dagger} = \mathcal{P}_n.$$ Elements of the Clifford group are generated by $\cnot$ (a two-qubit gate that maps $\ket{a,b}$ to $\ket{a,b \oplus a}$), Hadamard $\left(H = \frac{1}{\sqrt{2}}\begin{pmatrix} 1 & 1\\ 1 & -1 \end{pmatrix}\right)$, and Phase $\left(P = \begin{pmatrix} 1 & 0\\ 0 & i \end{pmatrix}\right)$ gates.

\paragraph{Universal Gate Set.} A set of gates is said to be universal if for any integer $n\geq 1$, any $n$-qubit unitary operator can be approximated to arbitrary accuracy by a quantum circuit using only gates from that set. It is a well-known fact that Clifford gates are not universal, but adding any non-Clifford gate, such as $T \left(T = \begin{pmatrix} 1 & 0\\ 0 & \sqrt{i} \end{pmatrix}\right)$, gives a universal set of gates. Throughout the paper, we will use the universal gate set $\{H,T,\cnot\}$.

\subsection{Quantum Query Algorithms}
We consider the quantum query model in this work, which gives quantum circuits access to some oracles. 
\begin{definition}[Classical Oracle]
\label{def:classic_oracle}
A classical oracle $\mathcal{O}$ is a unitary transformation of the form $U_f \ket{x,y, z} \rightarrow \ket{x, y+f(x), z}$ for classical function $f: \{0,1\}^n \rightarrow \{0,1\}^m$. Note that a classical oracle can be queried in quantum superposition.
\end{definition}
In the rest of the paper, unless specified otherwise, we refer to the word ``oracle'' as ``classical oracle''. 
d zdzA quantum oracle algorithm with oracle access to $\mathcal{O}$ is a sequence of local unitaries $U_i$ and oracle queries $U_f$. Thus, the query complexity of a quantum oracle algorithm is defined as the number of oracle calls to $\mathcal{O}$.

In the analysis of the security of the USS scheme in QROM (\Cref{thm:uss_qrom}), we will use the following theorem from \cite{BBBV97} to bound the change in adversary's state when we change the oracle's input-output behavior at places where the adversary hardly ever queries on. 

\begin{theorem}[\cite{BBBV97}] 
\label{thm:bbbv97_oraclechange}
Let $\ket{\phi_i}$ be the superposition of oracle quantum algorithms $\mathcal{M}$ with oracle $\cO$ on input $x$ at time $i$. Define $W_y(\ket{\phi_i})$ to be the sum of squared magnitudes in $\ket{\phi_i}$ of configurations of $\mathcal{M}$ which are querying the oracle on string $y$. For $\epsilon  > 0$, let $F \subseteq [0, T-1] \times \Sigma^*$ be the set of time-string pairs such that 
$\sum_{(i,y) \in F} W_y(\ket{\phi_i}) \leq \epsilon^2/T$.

Now suppose the answer to each query $(i, y) \in F$ is modified to some arbitrary fixed $a_{i,y}$ (these answers need not be consistent with an oracle). Let $\ket{\phi_i'}$ be the superposition of $\mathcal{M}$ on input $x$ at time $i$ with oracle $\cO$ modified as stated above. Then $\left\| \ket{\phi_T} - \ket{\phi_T'} \right\|_{\mathrm{tr}}\leq \epsilon$.
\end{theorem}

\subsection{Port-based Teleportation}
\label{sec:port_based_teleport}

In this section, we review a type of teleportation introduced by Ishizaka and Hiroshima \cite{ishizaka2008asymptotic}.
To distinguish their teleportation protocol from the traditional one, we borrow from
their terminology and call this port-based teleportation. 

The port-based teleportation protocol is described as follows: 
Alice wants to teleport a qudit state $\ket{\psi_{\mathbf A}}$ from her
system ${\mathbf A} \cong \mathbb{C}^d$
to Bob’s system ${\mathbf B} \cong \mathbb{C}^d$. We assume that Alice and Bob share $N = O(2^d)$ copies of the
maximally entangled state $\ket{\Phi} = \frac{1}{\sqrt{d}}\sum^d_i \ket{i}\ket{i}$
respectively in registers ${\mathbf A}'_1, {\mathbf B}'_1; {\mathbf A}'_2, {\mathbf B}'_2; \cdots {\mathbf A}'_N, {\mathbf B}'_N$. We
fix an orthonormal standard basis in each of these spaces. 

\begin{enumerate}
    \item Alice performs a certain POVM $\{E_{{\mathbf A}_i {\mathbf A}_i'}^i\}_{i = 1}^N $ on her systems $\{{\mathbf A}_i, {\mathbf A}_i'\}_{i\in [N]}$. She sends the result $i$ to Bob. 

    \item Bob discards everything except the subsystem ${\mathbf B}'_i$
and calls it ${\mathbf B}$. 

\item The guarantee of the protocol is that, this register ${\mathbf B}$ now holds the state $\ket{\psi_{\mathbf A}}$.
\end{enumerate}

\section{Definitions and Notations}


\subsection{Unclonable Secret Sharing}
An $(t,n)$-unclonable secret sharing scheme, associated with $n$ parties ${\cal P}_1,\ldots,{\cal P}_n$, consists of the following QPT algorithms: 
\begin{itemize}
    \item $\share(1^{\secparam},1^n, 1^t, m) \to \rho_{{\mathbf R}_1 {\mathbf R}_2 \cdots {\mathbf R}_n}$: On input security parameter $\secparam$, $n$ parties, a secret $m \in \{0,1\}^*$, output registers ${\mathbf R}_1, {\mathbf R}_2, \cdots, {\mathbf R}_n$.
    
    \item $\reconstruct(\rho_{{\mathbf R}'_{i_1}},\ldots,\rho_{{\mathbf R}'_{i_t}})$: On input shares ${\mathbf R}'_{i_1},\ldots,{\mathbf R}'_{i_t}$, output a secret $\widehat{m}$. 
\end{itemize}

When it is an $n$-out-of-$n$ USS scheme, we ignore the input $1^t$ in $\share$. 
In the rest of the work, we will focus on constructions with unentangled shares and impossibility results for entangled shared. For sake of clarity, we will use $\rho_1, \cdots, \rho_n$ to denote these shares. 
We require the following properties to hold. 

\paragraph{Correctness.} 
We can recover the secret with probability (almost) 1, more formally:
\begin{align*}
   \Pr[\reconstruct(\rho_{i_1}, \cdots, \rho_{i_k}) = m | (\rho_{1}, \cdots, \rho_{n}) \gets \share(1^\lambda, 1^n, m) \cap k \geq t] = 1-\negl(\lambda).
\end{align*}

\vipul{ I'm still not sure I understand why we don't get perfect correctness}
\qipeng{do we have imperfect construction? If not, we can change to perfect completeness.} \pnote{I don't think there is an issue in stating the def this way.}

\subsection{Indistinguishability-Based Security}
\label{def:uss_ind_security}
In this work, we will mostly focus on the $(n,n)$-unclonable secret sharing case. For simplicity, we call it $n$-party USS.

In this section, we define indistinguishability-based security for $n$-party $\USS$. The security guarantees that for any two messages $m_0, m_1$, no two reconstructing parties can simultaneously distinguish between whether the secret is $m_0$ or $m_1$, given their sets of respective cloned shares. Formally, we define the following experiment:\\

\noindent $\underline{\expt_{(\{\alice_i\},\bob,\charlie,\xi)}}$:
\begin{enumerate}
    \item Let $\xi$ be a quantum state on registers $\aux_1,\ldots,\aux_n$. For every $i \in [n]$,  $\alice_i$ gets the register $\aux_i$. 
    \item $\adversary=\left( \{\alice_i\},\bob,\charlie,\xi \right)$ sends $(m_0,m_1)$ to the challenger such that $|m_0|=|m_1|$.  \qipeng{do they need to have the same length?} \pnote{fixed it}
    \item \textbf{Share Phase:} The challenger chooses a bit $b \xleftarrow{\$} \{0,1\}$. It computes $\share(1^{\secparam},1^n, m_b)$ to obtain $\left(\rho_1,\ldots,\rho_n\right)$ and sends $\rho_i$ to $\alice_i$.  
    \item \textbf{Challenge Phase:} For every $i \in [n]$, $\alice_i$ computes a bipartite state $\sigma_{\regX_i \regY_i}$. It sends the register ${\regX_i}$ to $\bob$ and $\regY_i$ to $\charlie$. 
    \item $\bob$ on input the registers $\regX_1,\ldots,\regX_n$, outputs a bit $b_{\bob}$. $\charlie$ on input the registers $\regY_1,\ldots,\regY_n$, outputs a bit $b_{\charlie}$.
    \item Output 1 if $b_{\bob}=b$ and $b_{\charlie}=b$. 
\end{enumerate}

\noindent \pnote{Maybe we need to explicitly state that $\bob$ and $\charlie$ don't receive as any advice state that is entangled with $\xi$? I'm not sure if we have explicitly considered this setting before.}

\qipeng{why should we consider this? if the advice depends on the random oracle, this is the non-uniform oracle model; otherwise, any potential advice is included in $\xi$?}

\pnote{I see, so what you are saying that we can assume without loss of generality, Bob and Charlie don't get any advice state.. makes sense. I'll remove the comments.} 

\begin{definition}[Information-theoretic Unclonable Secret Sharing]
\label{def:ituss}
An $n$-party unclonable secret sharing scheme $(\share,\reconstruct)$ satisfies 1-bit unpredictability if for any non-uniform adversary $\adversary=\left(\{\alice_i\}_{i \in [n]},\bob,\charlie,\xi \right)$, the following holds:
$$\prob\left[1 \leftarrow \expt_{\left( \{\alice_i\},\bob,\charlie, \xi\right)} \right] \leq \frac{1}{2} + \negl(\secparam) $$
\end{definition}

\begin{definition}[Computational Unclonable Secret Sharing]
\label{def:compuss}
An $n$-party unclonable secret sharing scheme $(\share,\reconstruct)$ satisfies 1-bit unpredictability if for any non-uniform quantum polynomial-time adversary $\adversary=\left(\{\alice_i\}_{i \in [n]},\bob,\charlie,\xi \right)$, the following holds:
$$\prob\left[1 \leftarrow \expt_{\left( \{\alice_i\},\bob,\charlie,\xi \right)} \right] \leq \frac{1}{2} + \negl(\secparam) $$
\end{definition}
\qipeng{why this is non-uniform? We can just state QPT $\adversary$ and poly-sized $\xi$?} \pnote{not sure I see the problem}

\begin{claim}
\label{claim:n-1_imply_n_USS}
Existence of $(n-1)$-party $\USS$ unconditionally implies $n$-party $\USS$.
\end{claim}
\noindent This is straightforward to see, by creating a dummy share. 



\subsection{Entanglement Graph}


We will focus on the setting when there are multiple quantum adversaries with shared entanglement modeled as a graph, that we refer to as an {\em entanglement graph}. We formally define entanglement graphs below. 

\newcommand{\regx}{\mathcal{X}}
\begin{definition}[Entanglement Graph]
Let $\rho$ be a $n$-partite quantum state over the registers $\mathbf{X}_1, \cdots, \mathbf{X}_n$. Let $\rho[i]$ be the mixed state over register $\mathbf{X}_i$ $($i.e., $\rho[i]= \Tr_{\mathbf{X}_i}(\rho))$ and $\rho[i, j]$ be the mixed state over the registers $\mathbf{X}_i, \mathbf{X}_j$ $($i.e.,  $\rho[i,j]= \Tr_{\mathbf{X}_i,\mathbf{X}_j}(\rho))$. An entanglement graph $G = (V, E)$ associated with $(\rho,\mathbf{X}_1,\ldots,\mathbf{X}_n)$ is defined as follows:
\begin{itemize}
    \item $G$ is an undirected graph;
    \item $V = \{1, 2, \cdots, n\}$;
    \item $E$ contains an edge $(u, v)$ if and only if $\mathbf{X}_u$ and $\mathbf{X}_v$ are entangled; or in other words, there does not exist $\sigma_u, \sigma_v$ such that $\rho[u, v] = \sigma_u \otimes \sigma_v$.
\end{itemize}
\end{definition}


\noindent Performing non-local operations on a state $\rho$, over the registers $\mathbf{X}_1,\ldots,\mathbf{X}_n$, could change the entanglement graph. For instance, performing arbitrary channels on some $\mathbf{X}_i$, could remove some edges associated with the node $i$; for example, a resetting channel that maps every state to $\ket 0 \bra 0$. However, on the other hand, performing only unitary operations on each of $\mathbf{X}_1,\ldots,\mathbf{X}_n$ is not going to change the entanglement graph. 
\par Unless otherwise specified, we assume that the amount of entanglement shared between the different parties is either unbounded for information-theoretic protocols, or arbitrarily polynomial for computational protocols. 

\begin{definition}
Let ${\cal P}=({\cal P}_1,\ldots, {\cal P}_n)$ be the set of parties with $\rho$ being the state received by all the parties. That is, $\rho$ is an $n$-partite quantum state over the registers $\regX_1,\ldots,\regX_n$ such that the party ${\cal P}_i$ gets the register $\regX_i$. We say that $G$ is the entanglement graph associated with ${\cal P}$ if $G$ is the entanglement graph associated with $(\rho,\regX_1,\ldots,\regX_n)$.
\end{definition}

\begin{definition}[$\USS_d$]
\label{def:uss_n}
    We say an information-theoretic/computational unclonable secret sharing scheme is a secure $\USS_d$ scheme, if it has indistinguishability-based security against all unbounded/efficient adversaries with pre-shared entanglement, whose entanglement graph has \emph{at least} $d$ connect components.
\end{definition}
\noindent It is not hard to see that, $\USS_1$ is a USS satisfying the regular indistinguishability security.

\section{Adversaries with Disconnected Entanglement Graphs}

In this section, we give a construction of unclonable secret sharing with security against quantum adversaries with disconnected entanglement graphs.  
\subsection{\texorpdfstring{$\uss_{\omega(\log \secparam)}$: an Information-Theoretic Approach}{omega(log lambda) Connected Components: an Information-Theoretic Approach}}
\label{sec:IT-scheme-many-components}

\noindent \pnote{sketching the construction and proof..}

\noindent We present an information-theoretic protocol in the setting when there are $\omega(\log \lambda)$ connected components. For simplicity, we consider the case when there are $(n+1)$ parties and the entanglement graph does not have any edges. We demonstrate a construction of USS in this setting, where the security scales with $n$. \pnote{I'm considering sharing of bits below.. I believe we can generalize the argument to the setting when we have long messages?} \qipeng{seems not straightforward.} \pnote{agreed}

\begin{enumerate}
    \item $\share(1^{\secparam},1^{(n+1)},m \in \{0,1\})$: 
    \begin{enumerate}
        \item Sample uniformly random $r_1,\ldots,r_n \leftarrow \{0,1\}$ conditioned on $\oplus_{i} r_i = m$. 
        \item Sample $\theta_1,\ldots,\theta_n \leftarrow \{0,1\}$. 
        \item For each $i \in [n]$: let the $i^{th}$ share be $\rho_i = H^{\theta_i} \ketbra{r_i}{r_i} H^{\theta_i}$. Let the $(n+1)^{th}$ share be $\rho_{n+1}=(\theta_1,\ldots,\theta_n)$. 
        \item Output $\left(\rho_1,\ldots,\rho_{n+1} \right)$.
    \end{enumerate}
    \item $\reconstruct(\rho_1,\ldots,\rho_{n+1})$:
    \begin{enumerate}
        \item Measure $\rho_{n+1}$ in the computational basis to get $(\theta_1,\ldots,\theta_n)$.
        \item For every $i \in [n]$, apply $H^{\theta_i}$ to $\rho_i$. Measure the resulting state in the computational basis to get $r_i$.
        \item Output $\oplus_{i} r_i = m$. 
    \end{enumerate}
\end{enumerate}

\paragraph{Security.} Consider the adversary to be $\adversary=\left( \{\alice_i\},\bob,\charlie,\xi\right)$, where $\xi$ is a product state. Henceforth, we omit mentioning $\xi=\xi_1 \otimes \cdots \otimes \xi_{n+1}$, where $\alice_i$ receives $\xi_i$, since we can think of $\xi_i$ to be part of the description of $\alice_i$. 
\par For $b \in \{0,1\}$, let $\left(\rho_1^{r_1},\ldots,\rho_n^{r_n},\rho_{n+1} \right) \leftarrow \share(1^{\secparam},1^{(n+1)}, b)$, where $\oplus_i r_i = b$ and $\rho_i = H^{\theta_i} \ketbra{r_i}{r_i} H^{\theta_i}$ and $\rho_{n+1} = \ketbra{\theta_1 \cdots \theta_n}{\theta_1 \cdots \theta_n}$. Suppose upon receiving $\rho_i^{r_i}$, $\alice_i$ sends registers $\{{\bf X}_i^{r_i}\}$ and $\{{\bf Y}_i^{r_i}\}$ respectively to $\bob$ and $\charlie$. We denote the reduced density matrix on ${\bf X}_i^{r_i}$ to be $\sigma_i^{r_i}$ and on ${\bf Y}_i^{r_i}$ to be $\zeta_i^{r_i}$. We assume without loss of generality that $\rho_{n+1}$ is given to both $\bob$ and $\charlie$ since it is a computational basis state. 
\par Define ${\calS}_{\bob}$ and ${\calS}_{\charlie}$ as follows: 
$${\calS}_{\bob} =\left\{ i \in [n]\ : \td\left(\sigma_i^0,\sigma_i^1 \right) \leq 0.86 \right\} $$
$$ {\calS}_{\charlie} = \left\{ i \in [n]\ : \td\left(\zeta_i^0,\zeta_i^1 \right) \leq 0.86 \right\} $$
\par We prove the following claims. 

\newcommand{\ue}{\mathsf{ue}}
\begin{claim} 
\label{clm:bigboborbigcharlie}
Either $|\calS_{\bob}| \geq \lceil \frac{n}{2} \rceil$ or $|\calS_{\charlie}| \geq \lceil \frac{n}{2} \rceil$.
\end{claim} 
\begin{proof}
We prove by contradiction; suppose it is not the case. Then there exists an index $i \in [n]$ such that $i \notin \calS_{\bob}$ and $i \notin \calS_{\charlie}$. That is, $\td\left(\sigma_i^0,\sigma_i^1 \right) > 0.86$ and $\td\left(\zeta_i^0,\zeta_i^1 \right) > 0.86$, meaning the optimal state distinguishing circuit can distinguish $\sigma_i^0,\sigma_i^1$ with probability at least $0.93 = (1 + 0.86) / 2$. Similarly, the optimal distinguishing probability for states $\zeta_i^0,\zeta_i^1$  is at least $0.93$.

Using this, we design an adversary that violates the unclonable security of single-qubit BB84 states~\cite[Corollary 2]{BL20}.
Let us first recall the security game for the unclonability of single-qubit BB84 states: 
\begin{enumerate}
    \item $\alice$ receives $H^\theta \ketbra{x}{x} H^\theta$ for uniformly random $x, \theta \in \{0,1\}$, it applies a channel and produces $\sigma_{\mathbf{B} \mathbf{C}}$. Bob and Charlie receive their register accordingly. 
    \item Bob $\bob$ and Charlie $\charlie$ apply their POVMs and try to recover $x$; they win if and only if both guess $x$ correctly. 
\end{enumerate}
\begin{lemma}[Corollary 2 when $\lambda = 1$, \cite{BL20}]
    No (unbounded) quantum $(\alice, \bob, \charlie)$ wins the game with probability more than 0.855.
\end{lemma}
\noindent We design an adversary $(\alice, \bob, \charlie)$ as follows, with winning probability $0.86 > 0.855$, a contradiction.
\begin{itemize}
    \item $\alice$ receives as input an unknown BB84 state. It runs $\alice_i$ on the state to obtain a bipartite state, which it shares with $\bob$ and $\charlie$. 
    \item $\bob$ and $\charlie$ in the security game of BB84 state will receive $\theta_i$ from the challenger. 
    \item $\bob$ runs the optimal distinguisher distinguishing $\sigma_i^0$ and $\sigma_i^1$. Based on the output of the distinguisher, it outputs its best guess of the challenge bit. Similarly, Charlie runs the optimal distinguisher distinguishing $\zeta_i^0$ and $\zeta_i^1$. It outputs its best guess of the challenge bit. 
\end{itemize}
\noindent By a union bound, the probability that one of $\bob$ or $\charlie$ fails is at most $0.14 = 0.07 \times 2$. Thus, they simultaneously succeed with probability at least $0.86$, a contradiction.
\end{proof}

\begin{claim}
\label{clm:smalltd}
The following holds: 
\begin{enumerate}
\item 
$$\td\left(\sum_{\substack{r_1,\ldots,r_n:\\ \oplus_i r_i = 0}} \frac{1}{2^{n-1}} \left( \bigotimes_i \sigma_{i}^{r_i} \right),\ \sum_{\substack{r_1,\ldots,r_n:\\ \oplus_i r_i = 1}} \frac{1}{2^{n-1}} \left( \bigotimes_i \sigma_{i}^{r_i} \right)  \right) \leq  0.86^{|\calS_{\bob}|}$$
\item 
$$\td\left(\sum_{\substack{r_1,\ldots,r_n:\\ \oplus_i r_i = 0}} \frac{1}{2^{n-1}} \left( \bigotimes_i \zeta_{i}^{r_i} \right),\ \sum_{\substack{r_1,\ldots,r_n:\\ \oplus_i r_i = 1}} \frac{1}{2^{n-1}} \left( \bigotimes_i \zeta_{i}^{r_i} \right)  \right) \leq 0.86^{|\calS_{\charlie}|}$$
\end{enumerate}
\end{claim}
\begin{proof}
We prove bullet 1 since bullet 2 follows symmetrically. 
\begin{eqnarray*}
& & \td\left(\sum_{\substack{r_1,\ldots,r_n:\\ \oplus_i r_i = 0}} \frac{1}{2^{n-1}} \left( \bigotimes_i \sigma_{i}^{r_i} \right),\ \sum_{\substack{r_1,\ldots,r_n:\\ \oplus_i r_i = 1}} \frac{1}{2^{n-1}} \left( \bigotimes_i \sigma_{i}^{r_i} \right)  \right) \\
& = & \frac{1}{2} \left\| \sum_{\substack{r_1,\ldots,r_n:\\ \oplus_i r_i = 0}} \frac{1}{2^{n-1}} \left( \bigotimes_i \sigma_{i}^{r_i} \right) - \sum_{\substack{r_1,\ldots,r_n:\\ \oplus_i r_i = 1}} \frac{1}{2^{n-1}} \left( \bigotimes_i \sigma_{i}^{r_i} \right)  \right\|_1 \\
& = & \left\| \bigotimes_{i} \frac{\left( \sigma_i^0 - \sigma_i^1 \right)}{2}  \right\|_1 \\ 
& = & \prod_i \left\| \frac{\left( \sigma_i^0 - \sigma_i^1 \right)}{2} \right\|_1 \\
& \leq & \prod_{i \in \calS_{\bob}} \td\left( \sigma_i^0,\sigma_i^1 \right) \\
& \leq &  0.86^{|\calS_{\bob}|}
\end{eqnarray*} 
Here $\| \cdot \|_1$ denotes the trace norm. In the above proof, we use the fact that $\|\bigotimes_i \tau_i \|_1 = \prod_i \|\tau_i\|_1$. 
\end{proof}

\begin{lemma}
The above USS scheme satisfies indistinguishability security against any adversaries with no shared  entanglement; i.e., it is a secure $\uss_n$ scheme (see \Cref{def:uss_n}) with $n = \omega(\log \lambda)$. 
\end{lemma}
\begin{proof}
From~\Cref{clm:bigboborbigcharlie}, either $|\calS_{\bob}| \geq \lceil \frac{n}{2} \rceil$ or $|\calS_{\charlie}| \geq \lceil \frac{n}{2} \rceil$. We will assume without loss of generality that $|\calS_{\bob}| \geq \lceil \frac{n}{2} \rceil$. From bullet 1 of~\Cref{clm:smalltd}, it holds that $\bob$ can successfully distinguish whether it is in the experiment when the challenge bit 0 was used or when the challenge bit 1 was used, with probability at most $\frac{1+\nu(n)}{2}$, for some exponentially small function $\nu$ in $n$. Thus, both $\bob$ and $\charlie$ can only simultaneously distinguish with probability at most $\frac{1+\nu(n)}{2}$. This completes the proof. 
\end{proof}

\subsection{\texorpdfstring{$\uss_d$, for $d \geq 2$: from Unclonable Encryption}{ussd for d > 2: from Unclonable Encryption}}
\label{sec:UE_implies_USS2}


We present a construction of USS with security against quantum adversaries associated with \emph{any} disconnected entanglement graph. In the construction, we use an information-theoretically secure unclonable encryption scheme, $\UE = (\UE.\keygen, \UE.\enc, \UE.\dec)$. The resulting USS scheme is consequently information-theoretically secure. 

\begin{enumerate}
    \item $\share(1^\lambda, 1^n, m):$

    \begin{enumerate}
    
    \item Sample $r_1, \cdots,r_{n} \gets \{0,1\}^{|m|}$. 

\item  For each $i \in [n]$, let $y_i = r_i$; let $y_n = m \oplus \sum_{i=1}^{n} r_i$.
    \item For each $i \in [n]$:
    \begin{enumerate}
        \item Compute $\sk_i \gets \UE.\keygen(1^\lambda)$. We denote the length of $\sk_i$ to be $\ell=\ell(\secparam)$. 

        \item Compute $\ket{\ct_i} \gets \UE.\enc(\sk_i, y_i)$
    \end{enumerate}

    \item For each $i \in [n]$: let each share $\rho_i = (\sk_{i-1}, \ket{\ct_{i}})$; here we define $\sk_{0} = \sk_n$. 
    
    \item Output $(\rho_1, \cdots, \rho_n)$
    \end{enumerate}

    \item $\reconstruct(\rho_1,\cdots, \rho_n)$:
   
\begin{enumerate}
\item For each $i \in [n]$, 
    \begin{enumerate}
        \item Parse $\rho_i$ as $(\sk_{i-1}, \ket{\ct_{i}})$. We define $\sk_n = \sk_0$. 

        \item Compute $y_{i} \gets \UE.\dec(\sk_{i}, \ket{\ct_{i}})$
        
    \end{enumerate}
    \item Output $m = \sum_{i=1}^{n} y_i$.
    \end{enumerate}

\end{enumerate}

\begin{theorem}\label{thm:UE_implies_USS2}
The above scheme satisfies indistinguishability-based security against adversaries with any disconnected entanglement graph. More precisely, it is a secure $\USS_2$ scheme (see \Cref{def:uss_n}). 
\end{theorem}
\newcommand{\sh}{\mathsf{sh}}
\newcommand{\eps}{\varepsilon}
\begin{proof}
The correctness of the scheme follows from the correctness of \UE{} decryption. 

We now prove the security of the above scheme. Suppose we have an $\USS$ adversary $(\adv = (\adv_1, \cdots, \adv_n), \bob, \charlie, \xi)$ who succeeds with probability $\frac{1}{2} + \eps$ in~\Cref{def:uss_n}, we construct an $\UE$ adversary $(\adv',\bob',\charlie')$ who succeeds with probability $\frac{1}{2} + \eps$ in~\Cref{def:ue_ind}. 

Let $\adv$ receive as input an $n$-partite state $\xi$ over the registers $\aux_1,\ldots,\aux_n$ such that $\adv_i$ receives as input the register $\aux_i$. Additionally, without loss of generality, we can assume that $\adv$ also receives as input the challenge messages $(m_0,m_1)$, where $|m_0|=|m_1|$. Let $G=(V,E)$ be the entanglement graph associated with $(\xi,\aux_1,\ldots,\aux_n)$, where, $V=\{1,\ldots,n\}$. Since $G$ is disconnected, there exists $i^* \in [n]$ such that $(i^*,i^* + 1) \notin E$. Let $G_1=(V_1,E_1)$ and $G_2=(V_2,E_2)$ be two subgraphs of $G$ such that $V_1 \cup V_2 = V$, $V_1 \cap V_2 = \emptyset$, $i^* \in V_1$, $i^* +1 \in V_2$. Moreover, $G_1$ and $G_2$ are disconnected with each other. This further means that $\xi$ can be written as $\xi_{G_1} \otimes \xi_{G_2}$, for some states $\xi_{G_1},\xi_{G_2}$, such that $\xi_{G_1}$ is over the registers $\{\aux_i\}_{i \in V_1}$ and $\xi_{G_2}$ is over the registers $\{\aux_i\}_{i \in V_2}$.

We describe $(\adv',\bob',\charlie')$ as follows:

\paragraph{Description of $\adv'$.} Fix $i^*, (m_0, m_1)$ (as defined above).
Upon receiving a quantum state $\ket{\ct^*}$ 
$\adv'$ does the following: 
\begin{itemize}
    \item It prepares quantum states $\xi_{G_1},\left( \xi_{G_2} \right)^{\otimes 2^{\ell}}$. 
    \item It samples $r_i \xleftarrow{\$} \{0,1\}^{|m_0|}$, where $i \in [n]$, subject to the constraint that $\oplus_i r_i = m_0$. 
    \item It submits $(r_{i^*},r_{i^*} \oplus m_0 \oplus m_1)$ to the UE challenger and in return, it receives $\ket{\ct^*}$. It sets $\ket{\ct_{i^*+1}}=\ket{\ct^*}$. 
    \item For every $i \in [n]$, generate $\sk_i \gets \UE.\keygen(1^\lambda)$; let $\sk_{n+1} = \sk_1$. 
    \item For every $i \in [n]$ and $i \ne i^*$, generate $\ket{\ct_{i}} \leftarrow \UE.\enc(\sk_i,\sh_i)$. 
    \item For every $i \in [n]$ and $i \ne i^* + 1$, define $\rho_i = (\sk_{i-1}, \ket{\ct_i})$. 
    \item We need to define $\rho_{i^* + 1} = (\sk_{i^*}, \ket{\ct_{i^*+1}})$. However, as $\sk_{i^*}$ will only be received by $\bob'$ and $\charlie'$ in the UE security game later, we will enumerate all possible values of $\sk_{i^*}$ and the corresponding  computation result in the subgraph $G_2$. 
    \begin{itemize}
        \item For every $x \in \{0,1\}^{\ell}$ (possible value of $\sk_{i^*})$, compute $\{\adv_i\}_{i \in V_2}$ on $\{\rho_{i}\}_{i \in V_2}$, $\xi_{G_2}$ to obtain two sets of registers $\{\regX_i^{(x)}\}_{i \in G_2}$ and $\{\regY_i^{(x)}\}_{i \in G_2}$. 
    \end{itemize}
    
    \item Compute $\{\adv_i\}_{i \in V_1}$ on $\{\rho_i\}_{i \in V_1}$ and  $\xi_{G_1}$ to obtain two sets of registers $\{\regX_i\}_{i \in G_1}$ and $\{\regY_i\}_{i \in G_1}$. 
    
    \item Send the registers $\{\regX_i\}_{i \in G_1}$ and $\{\regX_i^{(x)}\}_{i \in G_2,x \in \{0,1\}^{\secparam}}$ to $\bob'$. Send the registers  $\{\regY_i\}_{i \in G_1}$ and $\{\regY_i^{(x)}\}_{i \in G_2,x \in \{0,1\}^{\secparam}}$ to $\charlie'$.  
\end{itemize}

\paragraph{Description of $\bob'$ and $\charlie'$.}  $\bob'$ upon receiving the secret key $k$ (which is $\sk_{i^*}$), computes $\bob$ on $\{\regX_i\}_{i \in G_1}$ and $\{\regX_i^{(k)}\}_{i \in G_2}$ to obtain a bit $b_{\bob}$. $\charlie'$ is defined similarly. We denote the output of $\charlie'$ to be $b_{\charlie}$. \\

\noindent If the challenger of the UE security chooses the bit $b=0$, then $(\adv,\bob,\charlie)$ in the above reduction are receiving shares of $m_0$; otherwise, they are receiving shares of $m_1$. Thus, the success probability of $(\adv,\bob,\charlie)$ in~\Cref{def:uss_n} is precisely the same as the success probability of $(\adv',\bob',\charlie')$ in~\Cref{def:ue_ind}. 
\end{proof}

\section{Adversaries with Full Entanglement}

\label{sec:construction_qrom}
\begin{theorem}(QROM protocol)\label{thm:uss_qrom}
There exists a $n$-party $\USS$ protocol with indistinguishability-based security  against adversaries sharing an arbitrary amount of entanglement ($\uss_1$, see \Cref{def:uss_n}) in the QROM, for any $n \geq 2$.
\end{theorem}

\paragraph{Construction}
Assume we have an underlying unclonable encryption scheme $\UE$ for one-bit messages (see \Cref{def:ue_ind}), consisting of three procedures $\UE.\keygen, \UE.\enc,\UE.\dec$ and a hash function $H:\{0,1\}^{k \cdot n} \to \{0,1\}^{\ell}$ modeled as a random oracle, where $\ell = \ell(\lambda)$ is the length of the $\UE$ secret key. It is easy to generalize our construction for the one-bit message to the $n$-bit message setting with indistinguishability based security in \Cref{def:uss_ind_security}.

We assume, without loss of generality, that the secret key generated from $\UE.\keygen$  is statistically close to uniform distribution\footnote{Given any UE scheme, we can convert it into another one where the setup outputs a random string. The new encryption algorithm will take this random string and runs the old setup to recover the secret key and then runs the old encryption algorithm.}.
We construct a $\USS$ scheme as follows:

\begin{itemize}
    \item $\share (1^\lambda, 1^n, m) \to (\rho_1, \cdots, \rho_n)$: 
    \begin{enumerate}
        \item Sample random $y_1, \cdots y_n \gets \{0,1\}^{k\cdot n}$, where $k = k(\lambda)$. Let $\sk = H(y_1, \cdots , y_n)$.

        \item Compute $\ket{\ct_m}$ = $\UE.\enc(\sk, m \in \{0,1\}$)
        
        \item  
        Let $\rho_1 = (\ket{\ct_m}, y_1)$; $\rho_2 = y_2$; $\rho_3 = y_3$; $\cdots; \rho_n = y_n$. 
    \end{enumerate}

    \item \reconstruct($\rho_1, \cdots, \rho_n$) $\to \hat{m}$: parse $\rho_1 = (\ket{\ct}, y_1)$; for every $i > 1$, measure $\rho_i$ to get $y_i$; compute $\sk = H(y_1, y_2 \cdots, y_n)$; compute $\hat{m} \gets  \UE.\dec(\sk, \ket{\ct})$.

\end{itemize}

\vipul{I think the above constrution should be collusion resistant: even if two disjoint sets of upto $n-1$ parties pool their shares and then clone, security should still hold.}
\jiahui{I agree. It seems equivalent to the two party case where we don't allow communication between them two. }

\paragraph{Correctness} The correctness of the above scheme follows from the correctness of the unclonable encryption scheme and of the evaluation of the random oracle $H$.

\subsection{Security} 
We consider the following two hybrids. We use \underline{underline} to denote the differences between Hybrid 0 and Hybrid 1. 

\paragraph{Hybrid 0.} the challenger operates the 1-bit unpredictability experiment according to the original construction above.
Let the adversary be $(\adv,\bob, \charlie)$ where $\adv = (\adv_1,\cdots, \adv_n)$.
\begin{enumerate}
     \item The challenger samples uniform random $y_1,\cdots,y_n \gets \{0,1\}^{n \cdot k}$  and  $ \sk \gets H(y_1,\cdots,y_n)$.

    \item The challenger samples secret $m \gets \{0,1\}$; computes $\ket{\ct_m} \gets  \UE.\enc(\sk,m)$.  Let $\rho_1 = (\ket{\ct_m}, y_1)$; $\rho_2 = y_2$; $\rho_3 = y_3$; $\cdots; \rho_n = y_n$.  

    \item The challenger gives the shares $\rho_1, \cdots, \rho_n$ to $\adv_1,\cdots, \adv_n$. 

    \item In the challenge phase, for every $i \in [n]$, $\alice_i$ computes a bipartite state $\sigma_{\regX_i \regY_i}$. It sends the register ${\regX_i}$ to $\bob$ and $\regY_i$ to $\charlie$. 

    $\bob$ on input the registers $\regX_1,\ldots,\regX_n$, outputs the bit $b_{\bob}$. $\charlie$ on input the registers $\regY_1,\ldots,\regY_n$, outputs the bit $b_{\charlie}$.
    \item The challenger outputs 1 if $b_{\bob}= m$ and $b_{\charlie} = m$. 
\end{enumerate}

\paragraph{Hybrid 1.} the challenger does the following modified version of the 1-bit unpredictability experiment: 
\begin{enumerate}
     \item The challenger samples uniform random $y_1,\cdots,y_n \gets \{0,1\}^{n \cdot k}$  and  \ul{$\sk \gets \{0,1\}^{\ell}$}, where $\ell = \ell(\lambda)$ is the length of the $\UE$ secret key.

    \item The challenger samples secret $m \gets \{0,1\}$; computes $\ket{\ct_m} \gets  \UE.\enc(\sk,m)$.  Let $\rho_1 = (\ket{\ct_m}, y_1)$; $\rho_2 = y_2$; $\rho_3 = y_3$; $\cdots; \rho_n = y_n$.

    \item The challenger gives the shares $\rho_1, \cdots, \rho_n$ to $\adv_1,\cdots, \adv_n$. \ul{It reprograms the random oracle $H$ at the input $(y_1, \cdots, y_n)$ to be $\sk$, right before entering the challenge phase. } In other words, the resulting random oracle $H'$ has the identical behavior as $H$ for every input except on input $(y_1, y_2, \cdots, y_n)$, $H'$ outputs $\sk$. 
     

    \item In the challenge phase, for every $i \in [n]$, $\alice_i$ computes a bipartite state $\sigma_{\regX_i \regY_i}$. It sends the register ${\regX_i}$ to $\bob$ and $\regY_i$ to $\charlie$. 

    $\bob$ on input the registers $\regX_1,\ldots,\regX_n$, outputs the bit $b_{\bob}$. $\charlie$ on input the registers $\regY_1,\ldots,\regY_n$, outputs the bit $b_{\charlie}$.
    \item The challenger outputs 1 if $b_{\bob}= m$ and $b_{\charlie} = m$. 
    
\end{enumerate}

\begin{lemma}
The advantages of adversary $\adv$ in Hybrid 0 and Hybrid 1 are negligibly close. 
\end{lemma}

\begin{proof}
We further clarify what happens in Hybrid 1: at the beginning of the execution, the function $H: \{0,1\}^{k\cdot n} \to \{0,1\}^{\ell}$ is a random function. Then it reprograms the function $H$ to get a new function $H'$ such that $H'(y_1, \cdots, y_n) = \sk$ right before the challenge phase; i.e., before any $\adv_i$ sends the bipartite state $\sigma_{\regX_i,\regY_i}$ to Bob and Charlie.

Suppose each $\adv_i$ has $q$ queries.  For each $\adv_i$, note that the strings $\{y_j\}_{j \neq i, j\in [n]}$ are uniformly random from  $\{0,1\}^{k(n-1)}$. Let us denote the squared amplitudes of $\adv_i^H$'s query  on input $x$ as $W_i(x)$.
By Markov's inequality, we have, for each $\adv_i$, given a fixed $y_i$ and for every $0 \leq \alpha \leq 1$:
\begin{align*}
    \Pr_{\{y_j\}_{j \neq i} \gets \{0,1\}^{k(n-1)}}[W_i(y_1,y_2, \cdots, y_n) \geq \alpha] \leq \frac{q}{\alpha \cdot 2^{k\cdot(n-1)}}
\end{align*}
\qipeng{Why there is a root? Need to recalculate the parameters here.}
\jiahui{fixed}
We can set $\alpha = \frac{1}{2^{(k\cdot(n-1))/2}}$ and have $ \Pr_{\{y_j\}_{j \neq i} \gets \{0,1\}^{k(n-1)}}[W_i(x) \geq  \frac{1}{2^{(k\cdot(n-1))/2}}] \leq \frac{q}{2^{(k\cdot(n-1))/2}}$.

Let us denote the query weight of the overall adversary $\adv = (\adv_1, \cdots, \adv_n)$ on input $(y_1, \cdots, y_n)$ as $W(y_1, \cdots, y_n)$.
Since the operations performed by $\adv_1, \cdots, \adv_n$ commute, we can apply the union bound to obtain the probability for any $\adv_i$ to query on $y_1, \cdots, y_n$:
\begin{align*}
        \Pr_{\{y_j\}_{j \neq i} \gets \{0,1\}^{k(n-1)}}[\exists i \in [n]: W_i(y_1,y_2, \cdots, y_n) \geq \alpha] \leq \frac{n\cdot q}{2^{(k\cdot(n-1))/2}}
\end{align*}
That is, with probability $(1 - \frac{n\cdot q}{2^{(k\cdot(n-1))/2}})$, for all $i \in [n]$, the query weight $W_i(y_1,\cdots, y_n) \leq \alpha$. Therefore, we have with overwhelmingly large probability, their query total weight $W(y_1, \cdots, y_n) = \sum_i W_i(y_1, \cdots, y_n) \leq n\cdot \alpha$.\pnote{this should be $\leq \alpha$?} \jiahui{this is the total weight of all parties I added some explanation}

We denote the joint state of $(\adv, \bob, \charlie)$ at the beginning of the challenge phase in Hybrid 0 as $\tau_{0}$
and the state in Hybrid 1 as $\tau_{1}$.
Then we can invoke~\Cref{thm:bbbv97_oraclechange}, to obtain $\left\|\rho_{\adv,0}  - \rho_{\adv,1} \right\|_{\mathrm{tr}} \leq \frac{\sqrt{n}\cdot q}{2^{k(n-1)/4}} = \negl(\lambda)$. We can view the final output in the challenge phase as the final outcome of a POVM on the state $\rho_{\adv,0}$ (or respectively, $\rho_{\adv,1}$). Therefore, by the fact that $\left\|\rho_{\adv,0}  - \rho_{\adv,1} \right\|_{\mathrm{tr}} = \max_E \left\|E(\rho_{\adv,0})  - E(\rho_{\adv,1}) \right\|_{\mathrm{tr}}$ where the maximum is taken over all POVMs $E$, we have $$\left|\Pr[(\adv,\bob,\charlie) \text{ wins Hybrid 0}] - \Pr[(\adv,\bob,\charlie) 
 \text{ wins Hybrid 1}]  \right| \leq \negl(\lambda).$$

\end{proof}

\begin{lemma}
Assuming the IND security of the unclonable encryption in \Cref{def:ue_ind}, the advantage of the adversary in Hybrid 1 is negligible.
\end{lemma}

\begin{proof}
Suppose $({\cal P}_1,\ldots,{\cal P}_n,\bob,\charlie)$ is an adversary in the security game of the unclonable secret sharing for the above construction, we construct a QPT $\adv'$ for
the indistinguishability unclonable encryption security defined in \Cref{def:ue_ind}.

$\adv'$ samples uniform random $y_1, \cdots, y_n$ and receives the quantum ciphertext $\ket{\ct_m}$ from the $\UE$ challenger. Then $\adv'$ prepares $\rho_1 = (\ket{\ct_m},y_1), \rho_2 = y_1, \cdots, \rho_n = y_n$ and sends them to $({\cal P}_1,\ldots,{\cal P}_n)$.

$\adv'$ prepares the state to send to $\bob', \charlie'$ as follows: gives $(y_1,\cdots, y_n)$ to  $\bob'$;
after entering the challenge phase of $\UE$ and before the challenge phase of $\USS$, $\bob'$ reprograms $H$ on input $y_1,\cdots, y_n$ to be $\sk$, which it receives from the $\UE$ challenger; then after entering challenge phase of $\USS$, $\bob'$ outputs the output of $\USS$ adversarial recoverer $\bob$ and $\charlie'$ outputs the output of $\USS$ adversarial recoverer $\charlie$. 

If $\bob,\charlie$ both outputs the correct $m$, then $\bob',\charlie'$ will win the IND $\UE$ security game. 
\end{proof}

\section{Impossibilities and Barriers}
In this section, we present two impossibility results on USS. Furthermore, we present two implications of USS: namely, unclonable encryption and position verification secure against large amount of entanglement. Since no construction known for the latter two primitives, this further underscores the formidable barriers of building USS. 

\subsection{Impossibility in the Information-Theoretic Setting}
\label{sec:general_impossibility}

\begin{theorem}
Let $\parties$ be a set of parties. Information-theoretically secure \USS\ for $\parties$ is impossible if the entanglement graph for $\parties$ is connected and in particular, there is an edge from $P_1$ to everyone else. 
\end{theorem}


\begin{proof}

The attack strategy is as follows. The $n$ parties $P_1, \cdots, P_n$ pre-share a large amount of entanglement with one another. In the protocol, each $P_i$ receives its share $\rho_i$. 

\begin{itemize}
    \item \emph{Regular Teleportation Stage}: all parties $P_i$, where $i \neq 1$  teleport their shares to party $P_1$ via regular teleportation.
    Each $P_i$ obtains a measurement outcome $(a_i, b_i)$.

    \item Now $P_1$ holds a state in the following format: $(\mathbb{I} \otimes X^{a_2} Z^{b_2} \otimes \cdots X^{a_n} Z^{b_n} )\ket{\Psi}_{P_1 P_2 \cdots P_n}$ which can be represented as mixed states $(\rho_1, X^{a_2} Z^{b_2} \rho_2 X^{a_2} Z^{b_2}, \cdots, X^{a_n} Z^{b_n} \rho_n X^{a_n} Z^{b_n})$. That is, quantum one-time padded shares from all other parties and its own share in the clear.

\item \emph{Port-Based Teleportation Stage}:
\begin{itemize}
  \item  $P_1$ performs port-based teleportation (see \Cref{sec:port_based_teleport}) for the state $(\mathbb{I} \otimes X^{a_2} Z^{b_2} \otimes \cdots X^{a_n} Z^{b_n} )\ket{\Psi}_{P_1 P_2 \cdots P_n}$ to $P_2$. $P_1$ obtains a measurement outcome that stands for some index $i_1$. Recall that by the guarantee of port-based teleportation, the index $i_1$  specifies the register of $P_2$ that holds the above state in the clear, \emph{without any Pauli errors on top}.
  
  \item $P_2$ will now remove the quantum one time pad information $X^{a_2}, Z^{a_2}$ on its share in the teleported state above. Since $P_2$ does not have information about $i_1$, it simply performs $\I \otimes Z^{a_2} X^{a_2} \otimes \I \cdots \otimes \I$ on all exponentially many possible registers that it may receive the teleported state from $P_1$.

  \item Next $P_2$ performs port-based teleportation with $P_3$ for \emph{all registers that could possibly hold the state} $(\mathbb{I} \otimes \mathbb{I}  \otimes X^{a_3} Z^{b_3} \otimes \cdots \otimes X^{a_n} Z^{b_n} )\ket{\Psi}_{P_1 P_2 \cdots P_n}$. Thus, $P_2$ obtains an exponential number of indices about the registers that will receive the teleported states on $P_3$'s hands.

  \item $P_3$ accordingly, applies $\I \otimes  \I  \otimes Z^{b_3} X^{a_3} \cdots \I$ on all the possible registers that can hold the teleported state; performs a port-based teleportation to $P_4$ with all of these registers and obtains a measurement outcome that has a doubly-exponential number of indices \footnote{For $P_i, 2 \leq i < n$, the measurement outcome will have its size grow in an exponential tower of height $i$. }.

  \item $\cdots$


  \item Finally, $P_n$ receives the teleported states from $P_{n-1}$ and performs $\I \otimes \cdots \I \otimes Z^{b_n} X^{a_n}$ on all of them. One of these registers will hold the state $\ket{\Psi}_{P_1\cdots P_n}= (\rho_1,\cdots, \rho_n)$ in the clear.
  Then $P_n$ performs the reconstruction algorithm on all of these registers to obtain a large number of possible outcomes. One of them will hold the correctly reconstructed secret $s$.
\end{itemize}

\item \emph{Reconstruction Stage}: now $P_n$ sends all its measurement outcomes to both Bob and Charlie. All other $P_i$'s send their indices information measured in the port teleportation protocol. Bob and Charlie can therefore find the correct index in $P_n$'s measurement outcomes that holds $s$, by following a path of indices.
\end{itemize}
\end{proof}

\begin{remark}
The above strategy can be easily converted into a strategy where the underlying entanglement graph is connected (but may not be a complete graph) and every pair of connected parties share (unbounded) entanglement. The similar idea applies by performing regular teleportation and port-based teleportation via any DFS order of the graph. Thus, we have the following theorem.
\end{remark}

\begin{theorem}
\label{thm:general_impossibility}
Let $\parties$ be a set of parties. Information-theoretically secure \USS\ for $\parties$ is impossible if the entanglement graph for $\parties$ is connected. 
\end{theorem}

\vipul{There might be some value in stating this theorem in a more general way: Any quantum secret sharing (for a classical message) can be ``locally" converted to a classical secret sharing: i.e., each party can apply a local procedure to convert it's share to be a purely classical share. This comes at the expense of exponential blowup in share size. Once we have this, implication to uncloneable secret sharing should follow as a corallary.}
\qipeng{Not fully get this. Could you give more explanation?}

\subsection{Impossibility with Low T-gates for Efficient Adversaries}
\label{sec:impossibility_lowT}
Our impossibility result above in the information-theoretic setting requires exponential amount of entanglement between the parties. In this section, we present an attack that can be performed by efficient adversaries, albeit on USS schemes with restricted reconstruction algorithms.

Again, we point out that our result is a rediscovery of a similar algorithm in \cite{speelman15} in the context of instantaneous non-local computation. We also extend the attack to an $n$-party setting whereas \cite{speelman15} considers only 2 parties.


\begin{theorem}
Let $\parties$ be a set of parties and if the entanglement graph for $\parties$ is connected, then there exists an attack using polynomial-time and polynomial amount of entanglement on any $\USS$ scheme where the procedure $\reconstruct$ consists of only Clifford gates and $O(\log\lambda)$ number of \T{} gates.
\end{theorem}


\begin{proof}
We first consider the two party case for the sake of clarity, and then generalize to $n$-party case.

Let us assume that the number of qubits in each party $\adv_i$'s secret share $\rho_i$ to be $k$  (up to some padding with $\ket{00\cdots 0}$ if they have different lengths) and they are each stored in registers $P_i$, respectively.
Without loss of generality, we view the entire system over registers $P_1$ and $P_2$ as a pure state $\ket{\psi_{12}}$ since our attack works regardless of this state being mixed or pure.

We write the honest protocol's reconstruction circuit $\reconstruct$ on $2k$-qubit inputs as a circuit consisting of Clifford gates and {\sf T} gates, followed by a measurement in the computational basis in the end. The secret $m$ to recover will be the first bit in the measurement outcome.

Recall that in the gate teleportation protocol \Cref{sec:quantum_gate_tp}, given the Pauli errors $(a,b)$ as the sender's (Alice) measurement result, and given a Clifford circuit $G$ the receiver (Bob) intends to apply on the teleported state $\ket{\psi}$, we can compute an update function  $f_G$ for $G$, so that $(a',b') \gets f_G(a,b)$ and $\Z^{b'}\X^{a'} G(\X^a \Z^b)\ket{\psi} = G\ket{\psi}$.
Instead of using the approach in \cite{BK21} to compute a (relatively complicated) update function for any Clifford+\T{} quantum circuit, we will instead use a simpler approach to compute the update function $f_\reconstruct$ just as for a Clifford-only circuit.

\begin{enumerate}
    \item  $\adv_1$ and $\adv_2$ pre-shares $k$-ebit of entanglement in register $P_1',P_2'$. $\adv_1$ teleports its share $\rho_1$ to $\adv_2$
    and obtains the measurement outcomes $(a, b) \in \{0,1\}^{2k}$. Now $\adv_2$ should have a quantum one-time padded $\X^a \Z^b \rho_1 \Z^b \X^a$ in its register $P_2'$.

    \item  $\adv_2$ applies  $\reconstruct$ circuit  in a gate-by-gate manner, with the following approaches:

    \begin{enumerate}
        \item  If the next gate to apply is a Clifford gate,  then $\adv_2$ simply applies it to the corresponding registers in $P_2'$ and/or $P_2$. Recall that $P_2'$ consists of teleported first share and $P_2$ has the second share.

        \item If the next gate is a \T{} gate, and suppose it is the $j$-th \T{} gate, according to the topological numbering on all $\T$ gates in the $\reconstruct$ circuit: $\adv_2$ first applies the \T{} gate. Then it samples a random bit $s_j \gets \{0,1\}$ and
        applies a $\Pgate^{s_j}$ gate upon applying the $j$-th $\T$ gate.
        (i.e., if $s_j = 1$ then it applies a $\Pgate$ gate upon applying the $\T$ gate,  and if $s_j = 0$, it does nothing).

        Every time after applying the $\T$ gate and its following $\Pgate^{s_j}$  gate, $\adv_2$ modifies the circuit description for $\reconstruct$: append the gate $\Pgate^{s_j}$ after the $j$-th \T{} gate; the gate $\Pgate^{s_j}$ operates on the same qubit. 
        
    \end{enumerate}

    \item In the end, $\adv_2$ finishes applying all the gates and obtains an modified reconstruction circuit $\reconstruct'$. It will obtain a classical outcome $c$.

    \item In the challenge phase, 
 $\adv_2$ sends the modified circuit description $\reconstruct'$ and $c$ to the recoverers $\bob$ and $\charlie$. $\adv_1$ sends the one-time pads $(a,b)$ to $\bob$ and $\charlie$.

    \item $\bob$ and $\charlie$ computes the update function $f_{\reconstruct'}$ according to the updated $\reconstruct'$
    circuit and computes $f_{\reconstruct'}(a||0\cdots 0, b ||0\cdots 0 )$ (the quantum OTPs $(a,b)$ are each appended with $k$ zeros to represent the quantum OTPs on $\adv_2$'s state $\rho_2$ before applying any gate).
    In the successful case, they will obtain Pauli corrections $(a^*, b^*)$; in an unsuccessful case, they abort \footnote{When computing the update function for the circuit $\reconstruct'$, 
we can modify the algorithm inside the update function to do the following:
if the next gate is a $\T$ gate and assume the update function after applying the previous gate gives outcome $(a',b')$, then first check if  there's a correct $\Pgate^{a'}$ gate following the \T{} gate; if yes, compute the update function for these two gates together as 
$(a', b'\oplus a') \gets f_{\Pgate^{a'} \T}(a', b')$; otherwise if there's not a correct $\Pgate^{a'}$ gate that follows the $\T$ gate, the update function aborts.}.
    They then each apply $\Z^{b^*}\X^{a^*}$ to $c$ and output the first bit of $\Z^{b^*}\X^{a^*}c$  as $m$ (in our settings, $\Z^{b^*}$ is in fact unnecessary).
\end{enumerate}

\paragraph{Correctness}
We show that the above strategy allows $\adv,\bob,\charlie$ to win with a noticeable probability, when the number of \T{} gates is $O(\log \lambda)$. 

Recall that we have the following identity (in \Cref{sec:update_function}):
\begin{align*}
         \T\left(\X^a\Z^b\right)\ket{\psi} & =  \left(\X^a \Z^{b\oplus a}\Pgate^a\right)\T\ket{\psi}
\end{align*}
Every time $\adv_2$ applies a $\T$ gate on the quantum one-time padded input, if we directly ``move'' this $\T$ gate to the right side, we will have an unwanted $\Pgate^a$ on the right side of the Pauli one-time pads $X,Z$'s. Our solution is to make a guess on the bit $a \in \{0,1\}$ and apply an additional $\Pgate^a$ in order to remove the unwanted phase gate $\Pgate^a$ from the final correction. 

When the guess is correct,  i.e. $s_j = a$, we have the following after applying the $\T$-gate and $\Pgate^{s_j} = \Pgate^a$ gate:
\begin{align*}
    \Pgate^a \X^a Z^{b \oplus a} \Pgate^a \T \ket{
    \psi} & = \X^a Z^{b \oplus a^2} \Pgate^{a \oplus a} \T \ket{
    \psi} \\
    & = \X^a \Z^{b \oplus a} \T \ket{
    \psi} \text{ since } a^2 =a, a\oplus a = 0.    
\end{align*}
The above equalities follow from applying the update rules in \Cref{sec:update_function}.

Now we have successfully "moved" the original $\T$ gate to the right side of Pauli pads $\X,\Z$'s, and removed the unwanted $\Pgate^a$ gate.

Let us suppose for all $j \in [\kappa]$, where $\kappa$ is the number of $\T$ gates, $\adv_2$'s guess for $s_j$ is correct:
that is, assuming the update function after applying the previous gate gives outcome $(a',b')$, then if $a' = 1$ then there is a $\Pgate$ gate that follows the $\T$ gate; else there would not be one. Then whenever the update function $f_{\reconstruct'}$ runs into a $\T$ gate followed by a correct $\Pgate^{s_j} = \Pgate^{a'}$, we will obtain the updated Pauli errors as $(a', b'\oplus a') \gets f_{\Pgate^{a'} \T}(a', b')$. 

If all $\adv_2$'s guesses for $\{s_j\}_{j \in [\kappa]}$ are correct, 
 then $\adv_2$'s measurement outcome $c = \reconstruct'\X^a\Z^b \ket{\psi_{12}}$ will be equal to $ \X^{a^*} Z^{b^*} \reconstruct \ket{\psi_{12}}$.
 $\bob$,  $\charlie$ will obtain the $(a^*, b^*) \gets f_{\reconstruct'}(a||0^k,b||0^k)$ so that by applying $\X^{a^*}\Z^{b^*} $ to $c$, they will both obtain the real result $\reconstruct\ket{\psi_{12}}$. by outputting the first bit they will recover the correct $m$.

Every time $\adv_2$ makes a guess on $s_j$, it has probability $\frac{1}{2}$ of getting correct. When it guesses incorrectly for one $j \in [\kappa]$, then the entire approach may fail. Therefore, the above attack has success probability at least $\frac{1}{2^\kappa}$. Since the $\T$ gate number $\kappa$ is logarithmic in terms of the security parameter, the attack succeeds with a noticeable probability.

\end{proof}

\paragraph{Extending to $n$-party case}
In the $n$-party case, every $\adv_i, i \neq n$ can teleport its share $\rho_i$ to the last party $\adv_n$ and sends their Pauli OTP information $\{(a_i, b_i)\}_{i \neq n}$ to the recoverers $\bob$ and $\charlie$. 

$\adv_n$ performs the same operations as what $\adv_2$ does in the 2-party case. In the end, it sends the outcome $c$ and modified circuit $\reconstruct'$ to $\bob$ and $\charlie$. They should then be able to compute corrections $(a_1',b_1', \cdots, a_n', b_n') \gets f_{\reconstruct'}(a_1',b_1', \cdots, 0^k, 0^k)$ and apply $\Z^{b_1', \cdots, b_n'}\X^{a_1',\cdots a_n'}$ to $c$ to recover the secret.

\subsection{USS Implies Unclonable Encryption}
\label{sec:USS_implies_UE}

\begin{theorem} \label{thm:USS_implies_UE}
Unclonable secret sharing with IND-based security against adversaries with (bounded) polynomial amount of shared entanglement and  connected pre-shared entanglement graph implies secure unclonable encryption.
\end{theorem}

\noindent We will first look at the 2-party case, which can be easily extended to the $n (> 2)$-party case.

\begin{proof}
Assume a secure $\USS = (\USS.\share, \USS.\reconstruct)$ with IND-based security, we construct the following \UE{} scheme:

\begin{enumerate}
    \item $\keygen(1^\lambda, 1^{|m|})$: samples a random $\sk \gets \{0,1\}^{2\ell}$, where $\ell = \ell(\lambda)$ is the number of qubits in each share generated by $\USS.\share(1^\lambda, 1^{|m|}, \cdot)$.
    Output $\sk$.

    \item $\enc(\sk,m):$ \begin{enumerate}
        \item compute $(\rho_{1}, \rho_{2}) \gets \USS.\share(1^\lambda, 1^{|m|}, m)$.

        \item sample random $(a,b) \gets \{0,1\}^{2\ell}$. Use them to quantum one-time pad the second share $\rho_2$ to obtain $\X^a\Z^b \rho_2 \Z^b\X^a$.

        \item compute $s \gets (a,b) \oplus \sk$
        
        \item Output $\ct = (\rho_{1}, \X^a \Z^b \rho_2 \Z^b\X^a, s)$.
    \end{enumerate}

    \item $\dec(\ct, \sk)$: 
    \begin{enumerate}
        \item parse $\ct = (\rho_1, \rho_2', s)$; 
        
        \item compute $(a,b) \gets s \oplus \sk$;
        
        \item output $m \gets \USS.\reconstruct(\rho_1, \X^{a} \Z^{b} \rho_2' \Z^b\X^a)$.
    \end{enumerate}
    
\end{enumerate}

\paragraph{Correctness} The correctness easily follows from the correctness of the underlying $\USS$ scheme.

\paragraph{Security}
Suppose we have $\UE$ adversaries $(\adv,\bob,\charlie)$ that wins in the IND-based $\UE$ security game, we can construct adversary $(\adv' = (\adv_1, \adv_2), \bob',\charlie')$ for the $\USS$ IND-based security.

Before receiving the shares from the challenger, $\adv_1$ and $\adv_2$ agrees on a random strong $r \gets \{0,1\}^{2\ell}$.
When receiving the shares,
$\adv_2$ teleports its share $\rho_2$ to $\adv_1$ and obtains Pauli errors $(a,b)$.

$\adv_1$ gives $(\rho_1, ,r)$ the $\UE$ adversary $\adv$.
$\adv_2$ computes $\sk' \gets (a,b) \oplus r$.

In the $\USS$ challenge phase, $\adv_2$ sends $\sk'$ to both $\bob'$ and $\charlie'$. The $\UE$ adversaries $\adv$ has finished giving the bipartite it genertaed from $(\rho_1, r)$ state $\sigma_{\bob, \charlie}$ to $\bob$ and $\charlie$.

Then $\bob'$ feeds $\bob$ with $\sk'$ as the secret key in the $\UE$ security game (and $\charlie'$ feeding $\sk'$ to $\charlie$,respectively), and outputs their output bit $b_\bob, b_\charlie$ as the answer to $\USS$ game. Since the classical part in the unclonable ciphertext is the classical information $(a,b)$ masked by a uniformly random $\sk$, the reduction perfectly simulates the above scheme by first giving the $\UE$ adversary $\adv$ a uniformly random string $r$ and later feeding $\bob,\charlie$ with $r \oplus (a,b)$.

\paragraph{Extending to $n$-party case} 
We can change the scheme to sample a longer $\sk \in \{0,1\}^{2(n-1)\ell}$ and
let the unclonable ciphertext be $(\rho_1, \X^{a_2}\Z^{b_2} \rho_2 \Z^{b_2}\X^{a_2}, \cdots, \X^{a_n}\Z^{b_n} \rho_n \Z^{b_n}\X^{a_n}, s = (a_1,b_1, \cdots, a_n, b_n) \oplus \sk )$.

In the reduction, when receiving the shares,
$\adv_i, i\neq 1$ teleports its share $\rho_i$ to $\adv_1$ and obtains Pauli errors $(a_i,b_i)$.
The rest of the reduction follows easily.

\end{proof}

\begin{theorem}
Unclonable secret sharing with IND-based security against adversaries with disconnected entanglement graph, where one of the parties receives as a share a quantum state and all other parties receive classical shares (in other words, computational basis states), implies secure unclonable encryption. 
\end{theorem}

\begin{proof}
In the case where only one party has a quantum share, the others classical shares, we can easily modify the above construction to have a $\UE$ scheme from $\USS$:

\begin{enumerate}
    \item $\keygen(1^\lambda, 1^{|m|})$: samples a random $\sk \gets \{0,1\}^{(n-1)\ell}$, where $\ell = \ell(\lambda)$ is the number of qubits/bits in each share generated by $\USS.\share(1^\lambda, 1^{|m|}, \cdot)$.
    Output $\sk$.

    \item $\enc(\sk,m):$ \begin{enumerate}
        \item compute $(\rho_{1}, y_{2},\cdots, y_n) \gets \USS.\share(1^\lambda, 1^{|m|}, m)$. $y_1, \cdots, y_n$ are binary strings.

        \item sample random $\sk \gets \{0,1\}^{(n-1)\ell}$.  Compute $s \gets (y_1, \cdots, y_n) \oplus \sk$
        
        \item Output $\ct = (\rho_{1}, s)$.
    \end{enumerate}

    \item $\dec(\ct, \sk)$: 
    \begin{enumerate}
        \item parse $\ct = (\rho_1, s)$; 
        
        \item compute $(y_1, \cdots, y_n) \gets s \oplus \sk$;
        
        \item output $m \gets \USS.\reconstruct(\rho_1, y_1, \cdots, y_n)$.
    \end{enumerate}
\end{enumerate}

\paragraph{Security}
Suppose we have an $\UE$ adversary $(\adv,\bob,\charlie)$ that wins in the IND-based $\UE$ security game with probability $\frac{1}{2} + \eps$, we construct an adversary $(\adv' = (\adv_1, \cdots \adv_n), \bob',\charlie')$ that wins in the $\USS$ IND-based security game with probability $\frac{1}{2} + \eps$. Thus, if the USS scheme is secure then $\eps$ has to be negligible. We describe $\adv_1,\cdots,\adv_n$ as follows.

Before receiving the shares from the challenger, $\adv_1, \cdots, \adv_n$ agrees on a random string $r \gets \{0,1\}^{(n-1)\ell}$.

$\adv_1$ gives $(\rho_1, r)$ to the $\UE$ adversary $\adv$.
$\adv_i$, for $i \neq 1$, when receiving the classical share $y_i$ from the challenger, computes $\sk_i' \gets y_i \oplus r_i$, where $r_i$ is the $(i-1)$-th block of length-$\ell$ string in $r$.

In the $\USS$ challenge phase, each $\adv_i$, for $i\neq 1$, sends $\sk_i'$ to both $\bob'$ and $\charlie'$. $\adv_1$ sends the bipartite state $\sigma_{\bob,\charlie}$ to $\bob'$ and $\charlie'$, where $\sigma_{\bob,\charlie}$ is the output of $\adv$. 

Then $\bob'$ feeds $\bob$ with $\sk' = (\sk_2', \cdots, \sk_n')$ as the secret key in the $\UE$ security game (and $\charlie'$ feeding $\sk'$ to $\charlie$, respectively), and outputs their output bit $b_\bob, b_\charlie$ as the answer to $\USS$ game. Since the classical part in the unclonable ciphertext is the classical information $(y_2,\cdots, y_n)$ masked by a uniformly random $\sk$, the reduction perfectly simulates the above scheme by first giving the $\UE$ adversary $\adv$ a uniformly random string $r$ and later feeding $\bob,\charlie$ with $r \oplus (y_2,\cdots, y_n)$. Thus, the advantage of $(\adv',\bob',\charlie')$ in breaking the USS security game is precisely the same as the advantage of $(\adv,\bob,\charlie)$ breaking the UE security game.

\end{proof}

\vipul{Is the above immidiate? Don't we need an argument that this ``equivalent" scheme after teleportation is still an uncloneable SS? If this works, wouldn't this end up showing that there is no hope of building uncloneable SS (w/o RO and with entanglement) unless we make progress on uncloneable encryption first? }
\jiahui{This would only imply a barrier for building USS against adversaries with connected-shared entanglement graph unless we build UE first. In the case where there's no connected entanglement, this conversion doesn't hold, so if we have >= two parties receiving quantum shares, they would not imply UE}

\subsection{Search-based USS Implies Position Verification} 
\label{sec:uss_imply_PV}

\qipeng{ Observations:
1. 2-party USS implies 1-dimensional PV. (seems harder than PV due to indistinguishability based security; while PV needs search security only.) 

2. $n$-party USS seems easier than $(n-1)$-dim PV. 
}

\noindent \pnote{it would be good to flesh this out.}

\paragraph{Quantum Position Verification}
We first give a definition of 1-dimensional quantum position verification.

A 1-dimensional quantum position-verification protocol in the vanilla model\footnote{The vanilla model is a model where we do not consider hardware restrictions on any parties. Usually it means that we do not work with bounded storage/bounded retrieval model; all parties have synchronized clocks. } of verifier ($V_0, V_1$) and prover $P$ consists of the following stages:
\begin{enumerate}
    \item Setup: Verifiers $V_0, V_1$ exchange information over another secure (possibly quantum) channel unknown to $P$ to prepare for the (potentially quantum) challenge $(\rho_x,\rho_y)$. $V_0, V_1$ also make sure that they are located on the two different sides of the prover $P$.
    
    \item Challenge: \begin{itemize}
        \item Verifiers $V_0$ sends $\rho_x$ to $P$ and $V_1$ sends $\rho_y$ to $P$ so that the two pieces of information reach $P$ at the claimed position $pos$ at the same time.
        
        \item $P$ computes $\cP(\rho_x,\rho_y)$ for some quantum channel $\cP$ instantaneously and sends the (possibly quantum) answers $\rho_{\ans,0}$ and  $\rho_{\ans,1}$ back to $V_0, V_1$.
        
        \item $V_0, V_1$ check if the answers arrive on the correct time  and if $\cP(\rho_x,\rho_y)$ is computed correctly. If both yes, accept; otherwise if one condition is violated, reject.
    \end{itemize}
\end{enumerate}
For any position $pos$ (within the capability of the verifiers), we want the protocol to satisfy two properties in terms of a security parameter $n$:
\begin{itemize}
    \item \textbf{Correctness}: For any honest prover $P$ at claimed position $pos$, there exists a negligible function $\negl(\cdot )$ such that the probability  that the verifiers accept is at least $(1 - \negl(\lambda))$.
    
    \item \textbf{Soundness}: For any malicious provers $(P_0, P_1, \cdots, P_k)$ (where $k = \poly(\lambda)$), none of which at the claimed position $pos$, there exists a negligible function $\negl(\cdot )$ such that the probability that the verifiers accept is at most $\negl(\lambda)$. 
\end{itemize}
where  the probabilities are taken over the randomness used in the protocol.

It can be observed that we can consider only two malicious provers $(P_0, P_1)$ since adding more provers won't help increase their winning probability.

\paragraph{QPV with Pre-shared Entanglement}
In QPV, we also consider different adversarial setup such as: (1) $(P_0, P_1)$ do not have pre-shared entanglement; (2) $(P_0, P_1)$ can share a bounded/unbounded polynomial amount of entanglement; 
(3) $(P_0, P_1)$ can share unbounded amount of entanglement.
We also divide the settings into computational and information-theoretic.

\begin{theorem}
\label{thm:uss_imply_pv}
   2-party $\USS$(computational/IT resp.) with search-based security (\Cref{def:uss_search}) implies 1-dimensional QPV (computational/IT, resp.), where the two adversarial provers in the QPV protocol pre-share the same amount of entanglement as the two parties in the $\USS$ protocol do.
\end{theorem}

\noindent The following theorem demonstrates from another point of view the barrier of constructing secure protocols against entangled adversaries for $\USS$ in the IT setting.
Even if we consider computational asusumptions, the development in building secure QPV protocols against entangled adversaries has been slow, which indicates further evidence on how challenging $\USS$ can be in the entangled setting.
\begin{theorem}[\cite{beigi2011simplified,buhrman2014position}]
\label{thm:pv_impossibility}
    Quantum position verification is impossible in the information theoretic setting if we allow the adversaries to preshare entanglement.
\end{theorem}

\paragraph{Proof for \Cref{thm:uss_imply_pv}}

\begin{proof}
    Given a 2-party $\USS$ protocol with search based security, we construct a QPV protocol as follows:

    \begin{enumerate}
        \item Setup: at time $t_0$, verifiers $V_0, V_1$ sample a random secret $s \gets \{0,1\}^n$.
        Run $\USS.\share(1^\lambda, 2, s) \to (\rho_0, \rho_1)$.
        
        \item $V_0$ sends $\rho_0$ to the prover and $V_1$ sends $\rho_1$ to the prover so that $\rho_0$ and $\rho_1$ reach the prover at the same time. Let us denote the time that these two messages arrive at prover's location as $t_1$.

        \item The prover runs $\USS.\reconstruct(\rho_0,\rho_1) \to s'$ (we assume that the reconstruction procedure is instantaneous compared to the travelling time of the message). It sends the recovered secret $s'$ to both $V_0$ and $V_1$.

        \item $V_0$ and $V_1$ check if the message $s'$ from the prover arrive on time, respectively and if the message $s' = s$. If yes, accept; else reject.
    \end{enumerate}
     Suppose there's a pair of malicious provers $(P_0, P_1)$ who are not at the claimed position but make the verfiers accept with non-negligible probability, then there exists some malicious shareholders $(P_0', P_1')$ that break the search-based security of 2-party $\USS$.

We first give a general description that captures all QPV attacks against the above protocol.
\begin{itemize}
    \item     $P_0$ will stand at a location between the left-hand-side verifier $V_0$ and the claimed location $pos$.   $P_1$ will stand at a location between the right-hand-side verifier $V_1$ and the claimed location $pos$.

    \item At time $t_{0,0}$, $P_0$ receives the message $\rho_0$ from $V_0$. $t_{0,0} < t_1$ because  $P_0$ is standing closer to $V_0$ than $pos$ is.

    \item At time $t_{0,1}$, $P_1$ receives $\rho_1$ from $V_1$. For similar reason as above, $t_{0,1} < t_1$.

    \item Without loss of generality, $P_0$ applies some unitary $U_0$ on $\rho_0$ and its own auxiliary state $\aux_0$ (possibly having preshared entanglement with $P_1$) to obtain two states $\rho_{0,L}, \rho_{0,R}$ as outputs.  

    \item $P_1$ applies some unitary $U_1$ on $\rho_1$ and its own auxiliary state $\aux_1$ (possibly having preshared entanglement with $P_0$) to obtain two states $\rho_{1,L}, \rho_{1,R}$ as outputs. 

    \item $P_0$ keeps $\rho_{0,L}$ for itself and sends $\rho_{0,R}$ to $P_1$; $P_1$ keeps $\rho_{1,L}$ for itself and sends $\rho_{1,R}$ to $P_0$.

    \item After $P_0$  receives $\rho_{1,L}$, it applies a POVM $\Pi_0$ on the joint system of $(\rho_{0,L}, \rho_{1,L})$ to get a measurement outcome $s'_0$.

    \item After $P_1$  receives $\rho_{0,R}$, it applies a POVM $\Pi_1$ on the joint system of $(\rho_{0,R}, \rho_{1,R})$  to get a measurement outcome $s'_1$ \footnote{ $U_0, U_1, \Pi_0, \Pi_1$ are all instantaneous compared to message travelling time.}.

    \item $P_0$ sends $s_0'$ to $V_0$; $P_1$ sends $s_1'$ to $V_1$.

\end{itemize}

 Suppose the above attack succeeds with probability $\epsilon$ \footnote{It is guaranteed that using the above attack strategy, the malicious provers' messages will arrive at the verifiers on time. We omit the details here since we do not need this property for attacking $\USS$.}, we construct a pair of $\USS$ adversary $(P_0', P_1')$ that succeeds with probability $\epsilon$ against search based security:

 \begin{itemize}
     \item $P_0'$ and $P_1'$ share the same setup (preshared entanglement/shared randomness) as $P_0, P_1$ do.

     \item When receiving share $\rho_0$, $P_0'$ applies the unitary $U_0$ on $\rho_0$ and its own auxiliary state $\aux_0$ (possibly having preshared entanglement with $P_1$) to obtain two states $\rho_{0,L}, \rho_{0,R}$.  

      \item When receiving share $\rho_1$, $P_1'$ applies the unitary $U_1$ on $\rho_1$ and its own auxiliary state $\aux_1$ (possibly having preshared entanglement with $P_0$) to obtain two states $\rho_{1,L}, \rho_{1,R}$.

      \item In the reconstruction stage: $P_0'$ sends $\rho_{0,L}$ to the recoverer $\bob$ and $\rho_{0,R}$ to $\charlie$.
       $P_1'$ sends $\rho_{1,L}$ to the recoverer $\bob$ and $\rho_{1,R}$ to $\charlie$.
    
      \item  $\bob$ applies the POVM $\Pi_0$ on the joint system of $(\rho_{0,L}, \rho_{1,L})$  to get a measurement outcome $s'_0$.
       $\charlie$ applies the POVM $\Pi_1$ on the joint system of $(\rho_{0,R}, \rho_{1,R})$ to get a measurement outcome $s'_1$.

       \item By our assumption, we have $s_0' = s_1' = s$ with probability $\epsilon$.
 \end{itemize}

\end{proof}

\printbibliography[heading=bibintoc]

\appendix

\section{Additional Preliminaries}

\subsection{Gate Teleportation Protocol}
\label{sec:quantum_gate_tp}

 Suppose Alice has a quantum state $\ket{\psi}$ (without loss of generality, $\ket{\psi}$ is a one qubit state) and Alice and Bob share one half each of an EPR pair. Then Alice can send her state to Bob using the quantum teleportation protocol. This requires Alice to perform a  measurement on the input as well as her half of the EPR pair (let the output of these measurements be $a,b$). Then Bob's part of EPR pair gets transformed to the state $X^b Z^a \ket{\psi}$. 

\par A simple modification to the quantum teleportation protocol allows us to achieve gate teleportation. That is, if we could apply $G$ to Bob's half of the EPR pair and then apply the quantum teleportation circuit, Bob gets the state $G\left(X^bZ^a\ket{\psi}\right)$.
 To obtain the correct outcome $G\ket{\psi}$, Bob needs to compute an update function  $f$ that helps him obtain $(a', b') \gets f(a,b)$ where the inputs are the Pauli correction $(a,b)$ ; Bob then applies $X^{a'} Z^{b'}$ on his register that hold $G\left(X^bZ^a\ket{\psi}\right)$.

\subsubsection{Update function}
\label{sec:update_function}

 We consider the following identities (ignoring the global phase) verbatim from~\cite{BK21}. Let $\ket{\psi}$ be a 1-qubit state and $\ket{\phi}$ be a 2-qubit state. 
As introduced in \cite{BK21}, we would like to obtain $G\left(\ket{\psi}\right)$ from $G\left(X^bZ^a\ket{\psi}\right)$ (output of gate teleportation). To look at how we can get $G\left(\ket{\psi}\right)$ from $G\left(X^bZ^a\ket{\psi}\right)$, we look at the case when $G$ is a 1-qubit gate and $\ket{\psi}$ is a 1-qubit state. For this we get the following identities (ignoring the global phase):
 \begin{align*}
     \X\left(\X^a\Z^b\right)\ket{\psi} & =  \left(\X^a\Z^b\right)\X\ket{\psi} \\
    \Z\left(\X^a\Z^b\right)\ket{\psi} & =  \left(\X^a\Z^b\right)\Z\ket{\psi} \\
   \Hgate\left(\X^a\Z^b\right)\ket{\psi} & =  \left(\X^b\Z^a\right)\Hgate\ket{\psi} \\
    \Pgate\left(X^aZ^b\right)\ket{\psi} & =  \left(\X^{a}\Z^{a \oplus b}\right)\Pgate\ket{\psi} \\
     \CNOT\left(\X^{a_1}\Z^{b_1}\otimes \X^{a_2}\Z^{b_2}\right)\ket{\phi} & =  \left(\X^{a_1}\Z^{b_1\oplus b_2}\otimes \X^{a_1\oplus a_2}\Z^{b_2}\right)\CNOT\ket{\phi} \\
     \T\left(\X^a\Z^b\right)\ket{\psi} & =  \left(\X^a \Z^{b\oplus a}\Pgate^a\right)\T\ket{\psi}
 \end{align*}
 
 From the above rules, we can see that for \emph{Clifford circuits}, we can prepare a classical update funtion $f_G$ according to the quantum circuit description $G$, apply $X^{a'}$ and $Z^{b'}$, where $(a',b') \gets f_G(a,b)$, to $G(X^bZ^a\ket{\psi})$ and get $G(\ket{\psi})$.

 \paragraph{Update function for Any Quantum Circuits}
 Additionally, as discussed in \cite{BK21}: suppose a state $\ket{\psi}$ is QOTP-ed using the keys $(a_1,b_1,\ldots,a_n,b_n)$. Then, for any quantum circuit $G$ (not necessarily Clifford) applied on the QOTP-ed state, there exists a correction unitary, expressed as a linear combination of Paulis, that when applied on the QOTP-ed state yields the state $G(\ket{\psi})$. Note that computing such an update function for arbitrary quantum circuits may not be efficient, depending on the number of \T{} gates.

\subsection{Search-Based Security and Collusion-Resistant Security}
\label{sec:uss_defs_search_collusion}

\paragraph{Search-Based Security}

In this paragraph, we briefly define search-based security, a weakening of the indistinguishability definition. The security guarantees that for a random message $m \gets \{0,1\}^\lambda$, no two reconstructing parties can simultaneously recover the secret $m$, given their set of respective cloned shares. 

The security game is the same as ~\cref{def:uss_ind_security} except that we replace the 1-bit mssage $b$ with the $\lambda$-bit message $m$.

Accordingly, the security definitions are:
\begin{definition}[Search-based Information-theoretic Unclonable Secret Sharing]
\label{def:uss_search}
An $n$-party unclonable secret sharing scheme $(\share,\reconstruct)$ satisfies search-based unpredictability if for any non-uniform adversary $\adversary=\left(\{\alice_i\}_{i \in [n]},\bob,\charlie,\xi \right)$, the following holds:
$$\prob\left[1 \leftarrow \expt_{\left( \{\alice_i\},\bob,\charlie, \xi\right)} \right] \leq \negl(\secparam) $$
\end{definition}

\begin{definition}[Search-based Computational Unclonable Secret Sharing]
An $n$-party unclonable secret sharing scheme $(\share,\reconstruct)$ satisfies search-based unpredictability if for any non-uniform QPT adversary $\adversary=\left(\{\alice_i\}_{i \in [n]},\bob,\charlie,\xi \right)$, the following holds:
$$\prob\left[1 \leftarrow \expt_{\left( \{\alice_i\},\bob,\charlie,\xi \right)} \right] \leq \negl(\secparam) $$
\end{definition}

\paragraph{Collusion-Resistant $\USS$ Security}
$t$-collusion resistant $\USS$ 
has the same security game as in \ref{def:uss_ind_security}, except that we allow an  adversarially and adaptively chosen partition of parties into size no larger than $t$ to collude: that is, in the share-phase, the shareholders can be partitioned into groups of size at most $t$, and within each group, the shareholders can communicate and output one bipartite state $\sigma_{\bob\charlie}$ to send $\bob,\charlie$ before the reconstruct stage. 
\vipul{ in addition, each group  outputs to registers: one for B and another for C}

\begin{claim}
\label{claim:2uss_imply_collusion}
    The existence of 2-party $\USS$ unconditionally implies $n$-party $\USS$ with any $t \, (t \leq n-1)$-collusion resistance.
\end{claim}

\begin{proof}
    First, let $k = \binom{n}{2}$.
    In the collusion resistance protocol, we first run a classical information theoretic $k$-out-of-$k$ secret sharing protocol on the original share $x$  to obtain classical shares $x_1,\cdots, x_k$.
    Next, for each
    every 2 party among all the $n$ parties for the collusion resistance adversary,
    run a 2-party $\USS$ protocol on secret $x_i, i\in [k]$. Thus, no matter how the adversary partitions the $n$ parties into groups, with each group having size $t < n$, there will always be at least a pair two parties apart in two different groups. 
    If the $n$-party collusion resistance protocol is insecure, then there exists a reduction that embeds a 2-party $\USS$ challenge in the secret sharing (by guessing correctly a separated pair of shareholders with inverse polynomial probability).
    The reduction can simulate the remaining $(k-1)$ runs of the two party $\USS$ protocols on its own. Suppose the collusion resistance adversary can output the correct secret in the end, then the reduction can recover the secret of the two party $\USS$ challenge as well.
    
    \end{proof}

\noindent \pnote{for eurocrypt, we might have to provide more details for the above claim: are we talking about indistinguishability or search security? what type of entanglement graphs would the collusion resistant protocol handle?}

\begin{corollary}
\label{cor:collusion_security_qrom}
 There exists a $n$-party $t$-collusion resistant $\USS$ protocol with indistinguishability-based security  against adversaries sharing an arbitrary amount of (connected) entanglement in QROM, for any polynomial $n \geq 2$ and any $t < n$.
\end{corollary}

\noindent The corollary easily follows from the above and \Cref{thm:uss_qrom}. 

\pnote{should we remove this? adding a remark might also be good}
\jiahui{will just write in a paragraph or two}

\jiahui{
Observation: 1.$n$ party USS implies $(n-1)$ party USS

2. $2$-party USS can be turned into n party USS with collusion resistance.
Our construction in QROM satisfies 2-party security thus collusion resistant
}
\vipul{We can also consider a stronger notion where up to t-1 parties collude: they pool their shares together and then output cloned shares. I haven't thought about it too much but it seems that our constructions might be generalized to achieve this stronger notion?}

\vipul{Self reminder: add some discussion about non-malleable and leakage-resiliant secret sharing models}


\end{document}